\documentclass[lettersize,journal]{IEEEtran}

\IEEEoverridecommandlockouts
\usepackage{amsthm,amsmath,amssymb}
\usepackage{caption}
\usepackage{graphicx}
\usepackage{stfloats}
\usepackage[font=small]{caption}
\usepackage{cite}
\usepackage{booktabs}
\usepackage[labelformat=simple]{subcaption}
\usepackage{soul}
\usepackage{tabularx}
\usepackage{textcomp}
\usepackage{mathtools}
\usepackage{color}
\usepackage[dvipsnames]{xcolor}
\usepackage[normalem]{ulem}

\newtheorem{definition}{Definition}

\newtheorem{remark}{Remark}
\newtheorem{lemma}{Lemma}
\newtheorem{theorem}{Theorem}

\usepackage{array}
\usepackage{algorithmic}
\usepackage{amsmath}

\usepackage{graphicx,float}

\usepackage[utf8]{inputenc}
\usepackage{enumerate}
\usepackage{color}
\usepackage{amsfonts}
\usepackage{algorithm}
\usepackage{algorithmic}

\def\BibTeX{{\rm B\kern-.05em{\sc i\kern-.025em b}\kern-.08em
    T\kern-.1667em\lower.7ex\hbox{E}\kern-.125emX}}

\begin{document}
	%
	\title{{Simultaneous Active and Passive Information Transfer for RIS-Aided MIMO Systems: Iterative Decoding and Evolution Analysis}}
	
	%
	%
	
	
	\author
	{
		\IEEEauthorblockA{Wenjun Jiang,~\IEEEmembership{Graduate Student Member,~IEEE}, and 
		Xiaojun Yuan,~\IEEEmembership{Senior Member,~IEEE}}

			\thanks{
			This work has been accepted for publication in IEEE Transactions on Signal Processing, 2023.

						
			W. Jiang and X. Yuan are with the National Key Lab. on Wireless Communi., University of Electronic Sci. and Tech. of China, Chengdu, China. (e-mail:wjjiang@std.uestc.edu.cn;  xjyuan@uestc.edu.cn).
			}
	}


	%
	%


	\maketitle
	
	%
	%
	%
	\begin{abstract}
		This paper investigates the potential of reconfigurable intelligent surface (RIS) for passive information transfer in a RIS-aided multiple-input multiple-output (MIMO) system. {We propose a novel simultaneous active and passive information transfer (SAPIT) scheme.} In SAPIT, the transmitter (Tx) and the RIS deliver information simultaneously, where the RIS information is carried through the RIS phase shifts embedded in reflected signals. We introduce the coded modulation technique at the Tx and the RIS.
		The main challenge of the SAPIT scheme is to simultaneously detect the Tx signals and the RIS phase coefficients at the receiver. To address this challenge, we introduce appropriate auxiliary variables to convert the original signal model into two linear models with respect to the Tx signals and one entry-by-entry bilinear model with respect to the RIS phase coefficients. With this auxiliary signal model, we develop a message-passing-based receiver algorithm. Furthermore, we analyze the fundamental performance limit of the proposed SAPIT-MIMO transceiver. Notably, we establish  state evolution to predict the receiver performance in a large-size system. We further analyze the achievable rates of the Tx and the RIS, which provides insight into the code design for sum-rate maximization. Numerical results validate our analysis and show that the SAPIT scheme outperforms the passive beamforming counterpart in achievable sum rate of the Tx and the RIS.
	\end{abstract}

	\begin{IEEEkeywords}
	Reconfigurable intelligent surface, multiple-input multiple-output, active and passive information transfer, message passing, state evolution.
	\end{IEEEkeywords}
	
	\section{Introduction}
	
	Reconfigurable intelligent surface (RIS), a.k.a. intelligent reflecting surface (IRS), has been envisioned as an emerging technology to empower sixth-generation (6G) wireless communications \cite{Yuan_survey,Liu_survey,Zheng_survey}. As a passive device, RIS can be deployed without requiring radio-frequency (RF) modules such as power amplifiers and analog-to-digital converters. A typical RIS consists of an array of reflecting elements, where each element can induce a  phase shift to the incident electromagnetic wave in a nearly passive manner. As such, the RIS was utilized as a passive beamformer to shape the wireless propagation channel \cite{Shuowen,Tianwei,Nemanja}. It was shown that through a collaborative optimization over the RIS phase shifts, RIS can provide a highly reliable link with the power gain quadratic to the number of the reflecting elements \cite{Renzo}. Yet, passive beamforming is not necessarily the most spectrum-efficient way of exploiting the ultimate potential of RIS.
	
	
	Recently, the use of RIS for passive information transfer has been studied in \cite{Wenjing_WC,Wenjing_JSAC,Ertugrul,Jing,Lechen,Shaoe_1,Shaoe_2,Shuaishuai,zhao2020metasurface}. As a pioneering attempt, the work in \cite{Wenjing_WC} studied passive beamforming and information transfer (PBIT) in a single-input multiple-output (SIMO) system, where the RIS delivers information through the random on/off state of each RIS element, referred to as on-off reflection modulation. This idea of PBIT was extended to the multiple-input multiple-output (MIMO) system in \cite{Wenjing_JSAC}, and a turbo message-passing (TMP) algorithm was proposed to alternatively detect the transmitter (Tx) signals and the RIS on/off states. 
    The authors in \cite{Shaoe_1} pointed out that the on-off reflection modulation at every RIS element leads to signal-to-noise-ratio (SNR) fluctuation. They alleviated this problem by designing appropriate reflection patterns of the RIS elements. Ref. \cite{Shaoe_2} further proposed quadrature reflection modulation to switch on all RIS elements, which avoids the SNR fluctuation problem. The authors in \cite{Shuaishuai} proposed a joint design of the RIS reflection patterns and the Tx signals for bit error rate (BER) minimization. {In \cite{zhao2020metasurface}, the authors built a RIS prototype that passively transfers information to multiple receivers in an indoor Wi-Fi scenario.} {However, a common problem in \cite{Ertugrul,zhao2020metasurface,Shuaishuai,Jing,Lechen,Shaoe_1,Shaoe_2} is that $n$-bit modulation of the RIS requires the use of $2^n$ different reflection patterns, which causes a high cost in reflection pattern design and related pattern detection, especially when the RIS operates at a high transmission rate.}
	
	 
	More recently, much research interest has been attracted to analyze the information-theoretic performance limit of simultaneous active and passive information transfer (SAPIT), where the Tx and the RIS deliver information simultaneously. In \cite{Karasik}, the authors studied SAPIT in the RIS-aided SIMO system and pointed out that joint channel encoding at the Tx and the RIS achieves a much higher achievable rate than the counterpart passive beamforming scheme. Ref. \cite{cheng2021degree} further extended the analysis to the RIS-aided MIMO system, and showed that passive information transfer outperforms passive beamforming in multiplexing gain. However, due to the high-dimensional integration in the computation of achievable rates, ref. \cite{Karasik} considered a small-scale system with a single Tx antenna and a few number of RIS elements, and ref. \cite{cheng2021degree}  considered asymptotic rate analysis in the high SNR regime.
	The study of RIS for passive information transfer is still in an infancy stage. Particularly, how to fully exploit the potential of RIS for passive information transfer with affordable complexity remains an open challenge.
	To address this challenge, we propose a new SAPIT scheme for the RIS-aided MIMO system. 
	{We first introduce coded modulation at the Tx and the RIS to relieve the burden of the reflection pattern designs as in \cite{Ertugrul,zhao2020metasurface,Shuaishuai,Jing,Lechen,Shaoe_1,Shaoe_2}. With coded modulation, every RIS element can deliver information by modulation techniques, e.g., phase shift keying (PSK) modulation, and the information is encoded to ensure reliable transmission.}
	Then, the main difficulty of the SAPIT scheme is to simultaneously retrieve the information from the Tx and the information from the RIS. This is a bilinear detection problem by noting that the received signal is a bilinear function of the Tx signal and the RIS phase coefficients. 
	We introduce appropriate auxiliary variables to convert the original signal model into three parts, namely, two linear models of the Tx signals, and one entry-by-entry bilinear model of the RIS phase coefficients. Based on this auxiliary system, we formulate the bilinear detection problem into a Bayesian inference problem, and then develop a message-passing algorithm.
	In the proposed algorithm, inferring the Tx signals from the two linear models is achieved by employing the approximate message passing (AMP) tool \cite{AMP_2009, sundeep,Parker_1}. More importantly, due to the fact that the entry-by-entry bilinear model of the RIS phase coefficients is among the simplest bilinear models, inferring the RIS phase coefficients is directly obtained through standard message passing. 
	
	We further analyze the fundamental performance limit of the SAPIT-MIMO transceiver. We show that in a large-size system, the receiver performance can be accurately characterized by state evolution (SE). A key finding in the SE is that in each message-passing iteration, the equivalent channels experienced by the Tx signals can be treated as additive Gaussian white noise (AWGN) channels and those experienced by the RIS phase coefficients can be treated as Rayleigh fading channels in the large-system limit. The equivalence to AWGN model is due to the use of the two linear models, which is similar to the finding in the SE analysis of AMP \cite{AMP_SE}. The equivalence to Rayleigh fading model is due to the use of the entry-by-entry bilinear model. {Based on the SE, we derive the achievable rates of the Tx and the RIS based on the mutual information and minimum mean-square error relationship \cite{Dongning2}.} Particularly, we interpret the message-passing process as the iteration between a bilinear detector and two decoders, and the maximal achievable sum rate is obtained when the transfer functions of the detector and the two decoders are matched. Simulations validate that SE results match well with numerical results, and show that the proposed SAPIT scheme significantly outperforms the counterpart passive beamforming scheme in terms of achievable sum rate.

	{
	The main contributions of this paper are summarized below:
	\begin{itemize}
		
		\item {We propose a novel SAPIT scheme in the RIS-aided MIMO system, where every RIS element can passively deliver information by the randomness of the phase coefficients. We introduce the coded modulation technique to relieve the burden of the existing reflection pattern designs.}
		
		
		\item {We show that the SAPIT-MIMO system model is equivalent to two linear models of the Tx signals plus one entry-by-entry bilinear model of the RIS phase coefficients. We then formulate the bilinear detection problem into a Bayesian inference problem, and develop a message-passing-based  detection algorithm.}
	
		\item {We analyze the fundamental performance limit of the proposed SAPIT-MIMO transceiver. Particularly, we establish the SE to accurately characterize the performance of the proposed massage-passing algorithm in the large-size system. Based on the SE, we derive the achievable-rate bounds of the Tx and the RIS.}
	   \end{itemize} 
	}

	{Passive beamforming aims to enhance the channel quality of the Tx by adjusting RIS reflection coefficients, while SAPIT aims to maximize the passive information transfer capability of the RIS by designing  RIS reflection patterns. The values of the RIS reflection coefficients are typically known to the receiver in passive beamforming, but unknown to the receiver in SAPIT. PBIT can be regarded as a combination of passive beamforming and SAPIT, where the RIS reflection coefficients are partially optimized to enhance the channel quality of the main link, and partially kept random for passive information transfer. By adjusting the amount of ``randomness'' in the reflection coefficients, there is generally a trade-off between the passive beamforming gain and the information delivery capability of the RIS. We show that, as two extremes of this trade-off, SAPIT can significantly outperform passive beamforming in terms of achievable sum rate of the overall scheme.
	
	}
	
	{\emph{Organization:} In Sec. II, we establish the SAPIT scheme in the RIS-aided MIMO system. In Sec. III, we formulate the receiver design problem into a Bayesian inference problem, and develop a message-passing algorithm. In Sec. IV, we develop state evolution to predict the algorithm performance and analyze the achievable rates of the Tx and the RIS. In Sec. V, we present extensive numerical results. In Sec. VI, we conclude this paper.}
	
	\emph{Notation:} We use bold capital letters (e.g., $\mathbf X$) for matrices and bold lowercase letters (e.g., $\mathbf x$) for vectors. $(\cdot)^T$, $(\cdot)^*$, and $(\cdot)^H$ denote the transpose, the conjugate, and the conjugate transpose, respectively. The cardinality of a set $\mathcal{X}$ is denoted as $|\mathcal{X}|$. We use ${\rm diag}(\mathbf x)$ for the diagonal matrix created from vector $\mathbf x$. $||\mathbf X ||_F$ and $||\mathbf x||_2$ denote the Frobenius norm of matrix $\mathbf X$ and the $l_2$ norm of vector $\mathbf x$, respectively. $\mathbf x_1 \prec \mathbf x_2$ represents that every coordinate of $\mathbf x_1$ is less than that of $\mathbf x_2$. $\delta(\cdot)$ denotes the Dirac delta function. Matrix $\mathbf I$ denotes the identity matrix with an appropriate size. For a random vector $\mathbf x$, we denote its probability density function (pdf) by $p(\mathbf x)$. The pdf of a complex Gaussian random vector $\mathbf x \in \mathbb C^N$ with mean $\mathbf m$ and covariance $\mathbf C$ is denoted by $\mathcal{CN}(\mathbf x ; \mathbf m , \mathbf C) = {\rm exp}(-(\mathbf x - \mathbf m)^H \mathbf C^{-1} (\mathbf x - \mathbf m))/(\pi^N |\mathbf C|)$. $\mathbb E[\cdot]$ and $\text{var}(\cdot)$ denote expectation and variance operators.

	\section{System Model}
	
	\subsection{Channel Model}
	 As illustrated in Fig. \ref{System}, we consider the RIS-aided MIMO system consisting of a $K$-antenna transmitter (Tx), an $M$-antenna receiver (Rx), and an $N$-element RIS. Denote by $\mathbf F \in \mathbb C^{N \times K}$, $\mathbf G \in \mathbb C^{M \times N}$, and $\mathbf H \in \mathbb C^{M \times K}$ respectively the baseband channel matrices from the Tx to the RIS, from the RIS to the Rx, and from the Tx to the Rx. We assume that the channel is block-fading, i.e., the channel matrices $\mathbf F$, $\mathbf G$, and $\mathbf H$ remain constant in a transmission block. Without loss of generality, each transmission block consists of $Q$ sub-blocks, and each sub-block is divided into $T \ge 1$ time slots. The Tx symbols are changed from time-slot to time-slot while the RIS phases are changed from sub-block to sub-block. In practice, the choice of $T$ is determined by the speed limit of the phase adjustment of the RIS controller \cite{zhang2018space}. We ignore the mutual coupling between RIS elements. The signals reflected by the RIS two or more times are also ignored due to severe path loss.\footnote{We note that mutual coupling and multi-reflection generally lead to more complicated physical models than (1) \cite{9856592,rabault2023tacit}. The corresponding SAPIT design for those models is a very interesting research topic, but is left for future work.} Then, the received signal model in the $t$-th time slot of the $q$-th sub-block is expressed as
	\begin{align}
		\label{model}
		\quad \mathbf{y}_{qt} & = \left( \mathbf G \text{diag}(\mathbf s_q) \mathbf F + \mathbf H \right) \mathbf x_{qt}  +\mathbf{w}_{qt}, ~ \forall q,t,
	\end{align}
	where $\mathbf y_{qt} = [y_{qt1},...,y_{qtM}]^T \in \mathbb{C}^{M}$ is the received signal vector in the $t$-th time slot of the $q$-th sub-block; $\mathbf{s}_q = [s_{q1},...,s_{qN}]^T \in \mathbb{C}^N $ and $\mathbf x_{qt} = [\mathbf x_{qt1},...,\mathbf x_{qtK}]^T \in \mathbb{C}^{K}$ are the reflecting phase coefficient vector of the RIS and the signal vector of the Tx, respectively; $\mathbf w_{qt}\in \mathbb{C}^{M}$ is an additive white Gaussian noise (AWGN) vector with elements independently drawn from $\mathcal{CN}(0,\sigma_w^2)$. Each element of $\mathbf{s}_q$ is constrained on a constellation $\mathcal{S}=\{ e^{j \theta_1},...,e^{j \theta_{|\mathcal{S}|}}\}$, where $\theta_i, i=1,...,|\mathcal S|$ are the adjustable phase angles of the RIS controller. 
	The constant modulus constraints of the RIS coefficients are naturally satisfied by our design. 
	We assume that each element of $\mathbf{x}_{qt}$ is constrained on a constellation $\mathcal{X}$ with cardinality $|\mathcal{X}|$ and unit average power. This assumption does not lose any generality since we can design the constellation points of $\mathcal X$ to approach any signal shaping (e.g., Gaussian signaling). We also assume that channel state information (CSI) is perfectly known. In practice, CSI can be acquired, e.g., by using the techniques in \cite{zhenqing} and \cite{Hang}. 
	
	\begin{figure}[h] 
		\vspace{-0.2 cm}
		\centering
		\includegraphics[width = 3.4in]{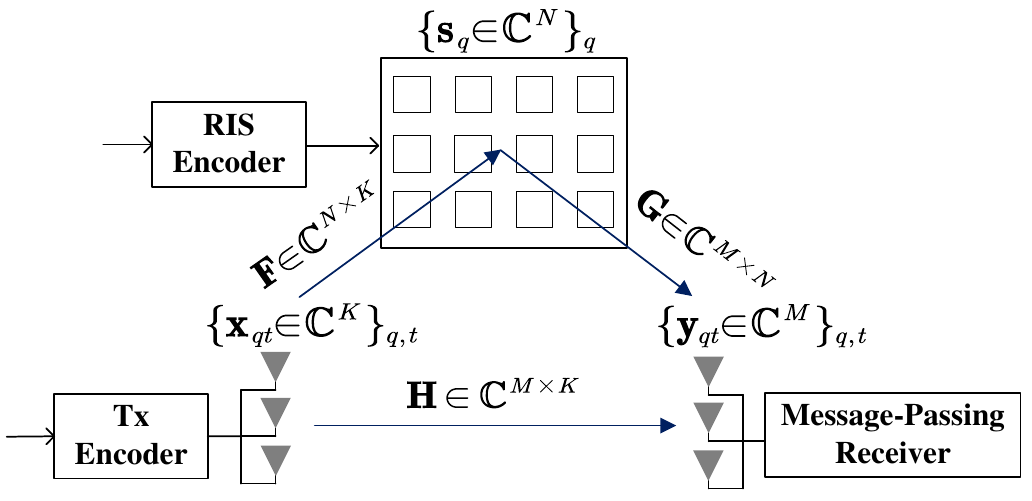}
		\vspace{-0.1 cm}
		\caption{Illustration of the RIS-aided MIMO system consisting of a $K$-antenna Tx, an $M$-antenna Rx, and an $N$-element RIS. Information are delivered by both the RIS and the Tx.}
		\label{System}
	\end{figure}
	
	\vspace{-0.2 cm}
	\subsection{Transmitter and RIS Design}
	
	In this subsection, we describe the transmitter and RIS design for SAPIT. As shown in Fig. \ref{System}, we introduce the Tx encoder and the RIS encoder for reliable transmission. Denote by $\mathcal C_x$ and $\mathcal C_s$ the codebooks adopted at the Tx and the RIS encoder, respectively. At the Tx, a codeword matrix $\mathbf X = [\mathbf x_{11},...,\mathbf x_{1T},...,\mathbf x_{Q1},...,\mathbf x_{QT}] \in \mathcal C_x $ is generated in each transmission block. At the RIS, we assign the first $N_{\rm P} \ll N$ rows of $\mathbf S = [\mathbf s_1, ... , \mathbf s_Q] \in \mathbb C^{N \times Q}$ as pilots known by the Rx, denoted by $\mathbf S_{\rm P} \in \mathbb C^{N_{\rm P} \times Q}$. Correspondingly, denote by $\mathbf S_{\rm D} \in \mathcal C_s$ the coded data matrix consisting of the remaining $N-N_{\rm P}$ rows of $\mathbf S$. Clearly, $ \mathbf S = [\mathbf S_{\rm P}^T,\mathbf S_{\rm D}^T]^T$. 
	{In practice, the data delivered by the RIS can be provided by connected internet-of-things (IoT) devices \cite{Wenjing_JSAC}. Alternatively, the data of the RIS can be split from the Tx source \cite{Karasik}.
	} 
	At the Rx, our goal is to simultaneously recover $\mathbf X$ and $\mathbf S_{\rm D}$ from the noisy observations $\mathbf Y = [\mathbf y_{11},...,\mathbf y_{1T},...,\mathbf y_{Q1},...,\mathbf y_{QT}]$ with the knowledge of $\mathbf S_{\rm P}$. This task can be formulated as a Bayesian inference problem, as detailed in the next section.

	\section{Message-Passing Receiver Design}
	
	In this section, we convert the receiver design problem into a Bayesian inference problem, and then develop a message-passing receiver algorithm. 
	\subsection{Bayesian Inference Framework}
	By introducing auxiliary random variables, system model \eqref{model} can be rewritten as
	\begin{subequations}
	\label{hybrid}
	\begin{align}
		& \mathbf y_{qt} = \mathbf G \mathbf u_{qt} + \mathbf H \mathbf x_{qt} + \mathbf w_{qt}, \label{hybrid_a} \\
		& \mathbf c_{qt} = \mathbf F \mathbf x_{qt}, \label{hybrid_b} \\
		& \mathbf u_{qt} = {\rm diag}(\mathbf s_q) \mathbf c_{qt}. \label{hybrid_c} 
	\end{align}	
	\end{subequations}
	It is worth noting that in the equivalent system model \eqref{hybrid}, inferring $(\mathbf u_{qt},\mathbf x_{qt})$ in \eqref{hybrid_a} and $\mathbf x_{qt}$ in \eqref{hybrid_b} are both linear inverse problems; the entry-by-entry multiplication of $s_{qn}$ and $c_{qtn}$ in \eqref{hybrid_c} is the simplest bilinear model. These properties facilitate the design of message-passing algorithm for the considered inference problem. {From \eqref{hybrid}, we have $p( y_{qtm}|\mathbf u_{qt},\mathbf x_{qt}) = \mathcal{CN}( y_{qtm} ; \sum_n g_{mn} u_{qtn} + \sum_k h_{mk} x_{qtk} , \sigma_w^2) $, $p(c_{qtn}|\mathbf x_{qt}) = \delta( c_{qtn} - \sum_k f_{nk} x_{qtk} )$, and $p( u_{qtn} | c_{qtn},s_{qn}) = \delta( u_{qtn} \!-\! s_{qn} c_{qtn})$. Define $\mathbf U \! = \! [\mathbf u_{11},...,\mathbf u_{1T},...,\mathbf u_{Q1},...,\mathbf u_{QT}]$ and $\mathbf C \!=\!$ $[\mathbf c_{11},...,\mathbf c_{1T},...,\mathbf c_{Q1},...,\mathbf c_{QT}]$. Then, the posterior distribution of $(\mathbf S_{\rm D},\mathbf X, \mathbf U, \mathbf C)$ given $\mathbf Y$ and $\mathbf S_{\rm P}$ is expressed as
	\begin{align}
		\label{p_post}
		& p(\mathbf S_{\rm D}, \mathbf X, \mathbf U, \mathbf C|  \mathbf Y,\mathbf S_{\rm P}) \notag \\
		& \quad \propto  \Big( \prod_{q,t,m,n} p(y_{qtm}|\mathbf u_{qt},\mathbf{x}_{qt}) p( u_{qtn}|  c_{qtn} , s_{qn} )p( c_{qtn} | \mathbf x_{qt}) \Big)   \notag \\
		& \hspace{1.3cm} \times p(\mathbf S_{\rm D}) p(\mathbf X) p(\mathbf S_{\rm P}),
	\end{align}
	}where $p(\mathbf S_{\rm D})$ is a uniform distribution over the RIS codebook $\mathcal C_s$ and $p(\mathbf X)$ is a uniform distribution over the Tx codebook $\mathcal C_x$. 
	
	We now represent \eqref{p_post} with the factor graph in Fig. \ref{FG}, where the conditional probabilities in \eqref{p_post} are represented by rectangular factor nodes and the random variables in \eqref{p_post} are represented by circular variable nodes. Modules A, B, and C in Fig. \ref{FG} correspond to models \eqref{hybrid_a}, \eqref{hybrid_b}, and \eqref{hybrid_c}, respectively. Based on the factor graph, we next develop a low-complexity message-passing algorithm. For notational convenience, denote by $\Delta_{a \to b}(b)$ the message from factor node $a$ to variable node $b$, and by $\Delta_{a \leftarrow b}(b)$ the message from variable node $b$ to factor node $a$. Denote by $\Delta_{b}(b)$ the message at node $b$ that combines all messages from edges connected to $b$.

	\begin{figure}[h] 
		\vspace{-0.2 cm}
		\centering
		\includegraphics[width = 3.4in]{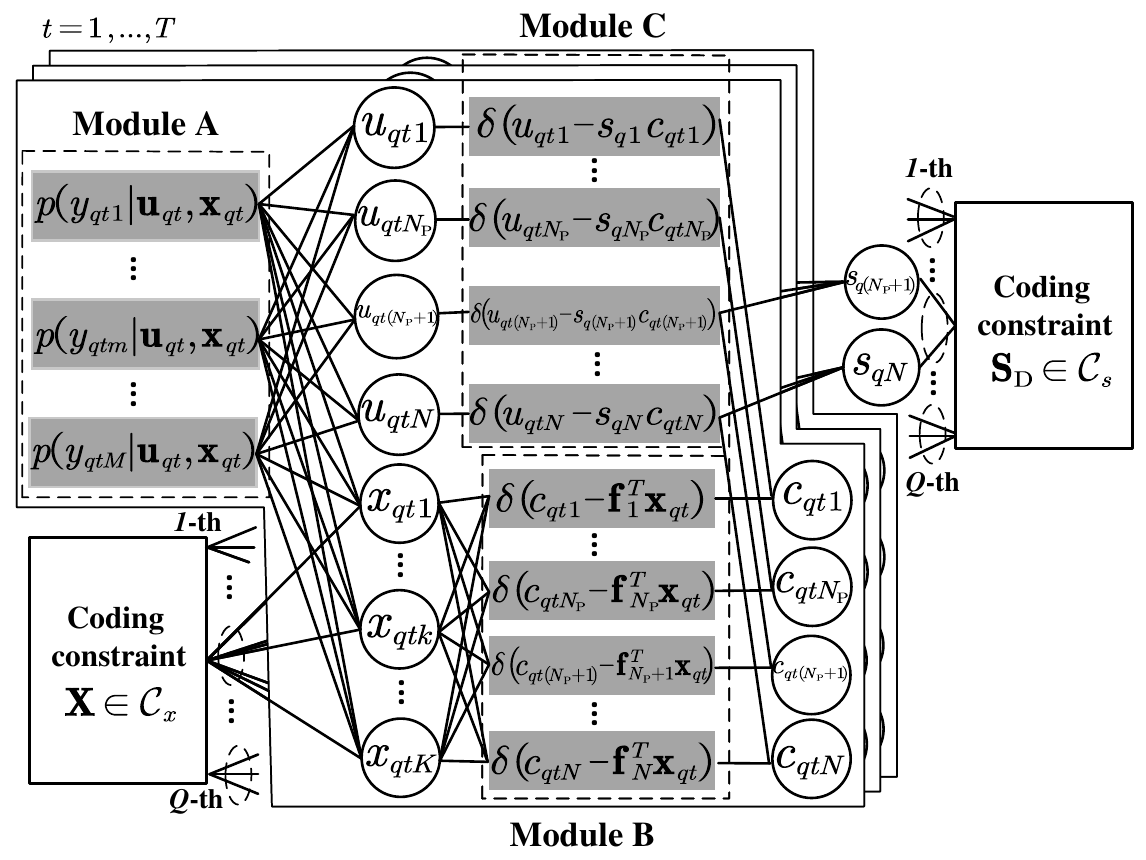}
		\vspace{-0.1 cm}
		\caption{Factor graph representation of \eqref{p_post}.}
		\label{FG}
	\end{figure}
	\subsection{Message-Passing Design}
	

	

	\subsubsection{Messages related to \{$u_{qtn}$\}}
	Define $ z_{qtm} = \mathbf g_m^T\mathbf u_{qt} + \mathbf h_m^T \mathbf x_{qt}$. From the sum-product rule \cite{sum-product}, the message from factor node $p(y_{qtm}|\mathbf u_{qt},\mathbf x_{qt})$ to variable node $u_{qtn}$ is
	\begin{align}
		\label{d_yu}
		\Delta_{y_{qtm} \! \to  u_{qtn}}(u_{qtn}) \! = \! \int_{z_{qtm}} \! \! \! \mathcal{CN}(y_{qtm};z_{qtm},\sigma_w^2) p(z_{qtm} | u_{qtn}),
	\end{align}  
	where $p(z_{qtm} | u_{qtn}) \!=\! \int_{\mathbf x_{qt},\mathbf{u}_t/u_{qtn}} \! \prod_k \! \Delta_{ y_{qtm} \leftarrow x_{qtk}}(x_{qtk}) \prod_{j \neq n}$ $\Delta_{y_{qtm} \leftarrow u_{qtj} }(u_{qtj}) p(z_{qtm} | \mathbf u_{qt}, \mathbf x_{qt}) \! $ with $\mathbf{u}_{qt}/u_{qtn}$ denoting the elements of $\mathbf u_{qt}$ except $u_{qtn}$ and subscript $y_{qtm}$ representing factor node $p(y_{qtm}|\mathbf u_{qt},\mathbf x_{qt})$. Note that $z_{qtm}$ is a weighted sum of random variables $u_{qtn},\forall n$ and $x_{qtk},\forall k$. From the central limit theorem, $p(z_{qtm} | u_{qtn})$ can be treated as a Gaussian distribution. Specifically, we use $\mu_{y_{qtm} \leftarrow x_{qtk}}$ and $v_{ y_{qtm} \leftarrow x_{qtk} }$ to  denote the mean and the variance of $ \Delta_{ y_{qtm} \leftarrow x_{qtk}  } (x_{qtk})$, and use $\mu_{ y_{qtm} \leftarrow u_{qtj} }$ and $v_{ y_{qtm} \leftarrow u_{qtj} }$ to denote the mean and the variance of $ \Delta_{ y_{qtm} \leftarrow u_{qtj}} (u_{qtj})$. Then, we obtain 
	\begin{align}
	p(z_{qtm} | u_{qtn}) \! = \! \mathcal{CN}(z_{qtm}; g_{mn}u_{qtn} \! + \! b_{ y_{qtm} \to u_{qtn} } , \tau_b{_{y_{qtm} \to u_{qtn}}} \! ), \label{pzu1}
	\end{align}
	with
	\begin{equation}
		\begin{aligned}
			& b_{y_{qtm} \to u_{qtn}}  \! = \! \sum_{j \neq n} g_{mj} \mu_{ y_{qtm} \leftarrow u_{qtj} } \! + \! \sum_k h_{mk} \mu_{y_{qtm} \leftarrow x_{qtk}}, \\
			& \tau_b{_{y_{qtm} \to u_{qtn}}} \! = \! \sum_{j \neq n} |g_{mj}|^2 v_{ y_{qtm} \leftarrow u_{qtj} } \! + \! \sum_k |h_{mk}|^2 v_{ y_{qtm} \leftarrow x_{qtk} }. \notag
		\end{aligned}
	\end{equation}
	Both $\{ b_{ y_{qtm} \to u_{qtn} } \}$ and $\{\tau_b{_{y_{qtm} \to u_{qtn}}}\}$ above have $NQTM$ elements to be updated in each iteration. For simplification, we further introduce quantities invariant to index $n$ as
	\begin{equation}
	\label{btaub}
	\begin{aligned}
		b_{qtm} & =\sum_{ n} g_{mn} \mu_{ y_{qtm} \leftarrow u_{qtn} } + \sum_k h_{mk} \mu_{y_{qtm} \leftarrow x_{qtk}}, \\
		\tau_{b_{qtm}} & = \sum_{n} |g_{mn}|^2 v_{ y_{qtm} \leftarrow u_{qtn} } + \sum_k |h_{mk}|^2 v_{ y_{qtm} \leftarrow x_{qtk} }.
	\end{aligned}
	\end{equation}
	Then, we have $	\tau_b{_{y_{qtm} \to u_{qtn}}}  = \tau_{b_{qtm}} -  |g_{mn}|^2 v_{ y_{qtm} \leftarrow u_{qtn} } \approx \tau_{b_{qtm}}$, where the approximation holds for a relatively large $N$. Denote by $\mu_{ u_{qtn} }$ the mean of message $\Delta_{u_{qtn}} (u_{qtn})$, and let $b_{ y_{qtm} \to u_{qtn} } = b_{qtm} - g_{mn} \mu_{ y_{qtm} \leftarrow u_{qtn} }\approx b_{qtm} - g_{mn} \mu_{ u_{qtn} }$. Substituting these approximations into \eqref{pzu1}, we obtain
	\begin{align}
		\label{pzu2}
		p(z_{qtm} | u_{qtn}) \approx \mathcal{CN}(z_{qtm}; g_{mn}(u_{qtn} - \mu_{ u_{qtn} }) + b_{qtm} , \tau_{b_{qtm}} ),
	\end{align}
	where we only need to update $\{ b_{qtm} \}$ and $\{ \tau_{b_{qtm}} \}$, each with $QTM$ elements. Substituting \eqref{pzu2} into \eqref{d_yu}, we obtain
	\begin{align}
		\label{d_yu2}
		\Delta_{y_{qtm} \to u_{qtn}}(u_{qtn}) \approx  \mathcal{CN} \big( & u_{qtn} ; \mu_{ u_{qtn} } + g_{mn}^{-1}(y_{qtm}-b_{qtm})\notag \\ & ~ |g_{mn}|^{-2}(\sigma_w^2+\tau_{b_{qtm}}) \big),
	\end{align}
	by following the Gaussian message combining property\footnote{$\mathcal{CN}(x ; a, A) \mathcal{CN}(x ; b, B) \propto \mathcal{CN}(x ; c, C)$ with $C=\left(A^{-1}+B^{-1}\right)^{-1}$ and $c=C(a / A+b / B)$.}.
	
	The message from variable node $u_{qtn}$ to factor node $p(y_{qtm}|\mathbf u_{qt},\mathbf x_{qt})$ is given by
	\begin{align}
		\label{dny}
		\Delta_{y_{qtm} \leftarrow u_{qtn} } \! (u_{qtn}) \! = \! \prod_{i \neq m} \! \Delta_{y_{qti} \to u_{qtn}}(u_{qtn}) \Delta_{\delta_{u_{qtn}} \to u_{qtn}} (u_{qtn})  .
	\end{align}
	With \eqref{d_yu2}, $\prod_{i \neq m} \Delta_{y_{qti} \to u_{qtn}}(u_{qtn})$ in \eqref{dny} is $\mathcal{CN}(u_{qtn};$ $d_{y_{qtm} \leftarrow u_{qtn}}, \tau_{d_{y_{qtm} \leftarrow u_{qtn}}}) $ with 
	\begin{equation}
		\label{d_taud_left}
	\begin{aligned}
		& \tau_{d_{y_{qtm} \leftarrow u_{qtn}}}  = \Big(\sum_{i \neq m} {|g_{in}|^2}({\sigma_w^2+\tau_{b_{qti}}})^{-1} \Big)^{-1}, \\
		& d_{y_{qtm} \leftarrow u_{qtn}}  = \mu_{ u_{qtn} } \! +  \! \tau_{d_{y_{qtm} \leftarrow u_{qtn}}} \sum_{i \neq m} g_{in}^* \frac{y_{qti} \! - \! {b}_{qti}}{\sigma_w^2\!+\!\tau_{b_{qti}}}.
	\end{aligned}
	\end{equation}
	$\Delta_{\delta_{u_{qtn}} \to u_{qtn}}(u_{qtn})$ in \eqref{dny} denotes the message from factor node $\delta(u_{qtn} - s_{qn}  c_{qtn})$ (abbreviated as $\delta_{u_{qtn}}$) to variable node $ u_{qtn}$. We next derive $\Delta_{\delta_{u_{qtn}} \to u_{qtn}}(u_{qtn}), n=N_{\rm P}+1,...,N$. Specifically, we have $
		\Delta_{\delta_{u_{qtn}} \to u_{qtn}}(u_{qtn}) \propto \int_{c_{qtn},s_{qn}} \! \delta(u_{qtn} \! - \! s_{qn}c_{qtn})  \Delta_{ \delta_{u_{qtn}} \!  \leftarrow s_{qn} } \! (s_{qn}) \Delta_{\delta u_{qtn} \leftarrow c_{qtn} } (c_{qtn}), n \! = \! N_{\rm P}+1,...,N$ with
	\begin{align}
		& \Delta_{ \delta_{u_{qtn}} \leftarrow s_{qn}} (s_{qn}) = \sum_{s \in \mathcal S} \pi_{qtn}(s) \delta(s_{qn} - s), \label{pri_s} \\
		& \Delta_{\delta_{u_{qtn}} \leftarrow c_{qtn} } (c_{qtn}) = \mathcal{CN}(c_{qtn} ; p_{qtn} , \tau_{p_{qtn}} ), \label{pri_c}
	\end{align}
	where $\pi_{qtn}(s)$ in $\Delta_{ \delta_{u_{qtn}} \leftarrow s_{qn}} (s_{qn})$ represents the probability of $s_{qn}=s$ for $s \in \mathcal{S}$. The detailed calculations of $\pi_{qtn}(s),n=N_{\rm P}+1,...,N$ are given later in \eqref{pi_s}. By noting $c_{qtn} = \sum_k f_{nk} x_{qtk}$ and the central limit theorem, $\Delta_{\delta u_{qtn} \leftarrow c_{qtn} } (c_{qtn})$ in \eqref{pri_c} is treated as a Gaussian distribution with mean $p_{qtn}$ and variance $\tau_{p_{qtn}}$. The detailed calculations of $p_{qtn}$ and $\tau_{p_{qtn}}$ are given later in \eqref{p_aprox}. With \eqref{pri_s} and \eqref{pri_c}, $\Delta_{\delta_{u_{qtn}} \to u_{qtn}}(u_{qtn}), n=N_{\rm P}+1,...,N$ is obtained as
	\begin{align}
		\label{ddu2}
		\Delta_{\delta_{u_{qtn}} \! \to u_{qtn}} \! (u_{qtn}) \! = \! \sum_{s \in \mathcal S} \pi_{qtn} \! (s) \mathcal{CN}(u_{qtn} ; p_{qtn} s, \tau_{p_{qtn}} \!).
	\end{align}
	Recall that $s_{qn},n=1,...,N_{\rm P}$ are pilots. Then, we obtain $\Delta_{\delta_{u_{qtn}} \to u_{qtn}}(u_{qtn}), n=1,...,N_{\rm P}$ by setting $\pi_{qtn}(s) = 1 $ for constellation point $s = s_{qn}$ in \eqref{ddu2}. Substituting \eqref{d_taud_left} and \eqref{ddu2} into \eqref{dny}, we obtain $\Delta_{y_{qtm} \leftarrow u_{qtn} }(u_{qtn}) = (\sum_{s \in \mathcal S} \pi_{qtn}  (s) \mathcal{CN}(u_{qtn} ; p_{qtn} s, \tau_{p_{qtn}} )) \mathcal{CN}(u_{qtn}; d_{y_{qtm} \leftarrow u_{qtn}}, $ $\tau_{d_{y_{qtm} \leftarrow u_{qtn}}}) $. We next  express the mean and the variance of $\Delta_{y_{qtm} \leftarrow u_{qtn} }(u_{qtn})$ as
	\vspace{-0.3cm}
	\begin{subequations}
	\label{uv_u}
	\begin{align}
		\label{u_hat}
		 \mu_{ y_{qtm} \! \leftarrow u_{qtn} }  &  \! = \! \int_{u_{qtn}} \! u_{qtn} \Delta_{y_{qtm} \leftarrow u_{qtn} }(u_{qtn}) \notag \\
		 & \! = \!  f_{qtm} ( d_{y_{qtm} \leftarrow u_{qtn}} \! ; \! \tau_{d_{y_{qtm} \leftarrow u_{qtn}}} ), \\
		v_{ y_{qtm} \! \leftarrow u_{qtn} } & \! = \! \int_{u_{qtn}} \! |u_{qtn} \! - \!\mu_{ y_{qtm} \leftarrow u_{qtn} } |^2  \Delta_{y_{qtm} \leftarrow u_{qtn} }(u_{qtn}) \notag \\
		& \! = \! 
		 \tau_{d_{y_{qtm} \! \leftarrow u_{qtn}}} \! f^{'}_{qtm} \! ( d_{y_{qtm} \! \leftarrow u_{qtn}} \! ; \! \tau_{d_{y_{qtm} \! \leftarrow u_{qtn}}} \! ), \label{v_u_hat}
	\end{align}
	\end{subequations}
	where $\mu_{ y_{qtm} \leftarrow u_{qtn} }$ in \eqref{u_hat} is regarded as a function of $d_{y_{qtm} \leftarrow u_{qtn}}$, and the derivations involved in \eqref{v_u_hat} are similar to those in \cite[eq. (43)-(44)]{Parker_1}. In \eqref{uv_u}, $\{ d_{y_{qtm} \leftarrow u_{qtn}} \}$ and $\{ \tau_{d_{y_{qtm} \leftarrow u_{qtn}}} \}$ have $MNTQ$ elements to be updated. We introduce quantities invariant to index $m$, i.e., the mean and variance of $\prod_{m} \Delta_{y_{qtm} \to u_{qtn}}(u_{qtn})$ as
	\begin{subequations}
	\label{y_to_u}
	\begin{align}
		\label{tau_qnt}
		& \tau_{d_{qtn}} = \Big(\sum_{m}  |g_{mn}|^2 (\sigma_w^2+\tau_{b_{qtm}})^{-1} \Big)^{-1}, \\
		& d_{qtn} = \mu_{ u_{qtn} } + \tau_{d_{qtn}} \sum_{m} g_{mn}^* (\sigma_w^2+\tau_{b_{qtm}})^{-1} (y_{qtm}-b_{qtm}), \notag
	\end{align}
	\end{subequations}
	with $\tau_{d_{y_{qtm} \leftarrow u_{qtn}}} \! \approx \! \tau_{d_{qtn}}$ and $d_{y_{qtm} \leftarrow u_{qtn}} \! \approx \! d_{qtn} - \tau_{d_{qtn}}  g_{mn}^* \times $ $ (\sigma_w^2+\tau_{b_{qtm}})^{-1} (y_{qtm}-b_{qtm})$. Substituting these approximations into \eqref{u_hat} and applying the first-order Taylor series expansion at point $d_{qtn}$, we obtain
	\begin{align}
		\label{u_ext}
		\mu_{ y_{qtm} \leftarrow u_{qtn} } & \approx f_{qtm}(d_{qtn} ; \tau_{d_{qtn}}) \notag \\ & 
		~ \quad - f'_{qtm}( d_{qtn} ; \tau_{d_{qtn}} ) \frac{\tau_{d_{qtn}} g_{mn}^* (y_{qtm}-b_{qtm})}{\sigma_w^2+\tau_{b_{qtm}}}. 
	\end{align}
	From the definitions of $f_{qtm}(\cdot)$ and $f'_{qtm}(\cdot)$ in \eqref{uv_u} and the sum-product rule, we obtain $ f_{qtm}(d_{qtn} ; \tau_{d_{qtn}})$ and $f'_{qtm}( d_{qtn} ; \tau_{d_{qtn}} ) \tau_{d_{qtn}}$ as the mean and the variance of $\Delta_{u_{qtn}} (u_{qtn}) = \prod_{m} \Delta_{y_{qtm} \to u_{qtn}} (u_{qtn}) \Delta_{\delta_{u_{qtn}} \to u_{qtn}} (u_{qtn})  $, denoted by $\mu_{u_{qtn}}$ and $v_{u_{qtn}}$, respectively. We express $\mu_{ u_{qtn} }$ and $v_{u_{qtn}}, n = N_{\rm P}+1,...,N$ as
	\begin{subequations}
	\begin{align}
		\mu_{ u_{qtn} } & = \mathbb{E}[u_{qtn}|d_{qtn},p_{qtn}, s_{qn} \sim \pi_{qtn};\tau_{d_{qtn}},\tau_{p_{qtn}}], \label{mu_u} \\  
		v_{u_{qtn}} & = {\rm var}(u_{qtn}|d_{qtn},p_{qtn},s_{qn} \sim \pi_{qtn};\tau_{d_{qtn}},\tau_{p_{qtn}}), \label{v_u_qtn} 
	\end{align} 
	\end{subequations}
	where $ \mathbb{E}[\cdot]$ and ${\rm var}(\cdot)$ are with respect to $p(u_{qtn}|d_{qtn},p_{qtn}, s_{qn} \sim \pi_{qtn};\tau_{d_{qtn}}, \tau_{p_{qtn}}) \propto \big( \sum_{s \in \mathcal S} \pi_{qtn}(s)$ $\mathcal{CN}(u_{qtn} ; p_{qtn} s, \tau_{p_{qtn}}) \big) \mathcal{CN}(u_{qtn};d_{qtn},\tau_{d_{qtn}})$ with $s_{qn} \sim \pi_{qtn}$ being the abbreviation of $s_{qn} \sim \sum_{s \in \mathcal S} \pi_{qtn}(s) \delta(s_{qn} - s)$. $\mu_{ u_{qtn} }$ and $v_{u_{qtn}}, n = 1,...,N$ are obtained by setting $\pi_{qtn}(s) = 1 $ for constellation point $s = s_{qn}$ in \eqref{mu_u} and \eqref{v_u_qtn}, respectively.

	\subsubsection{Messages related to \{$s_{qn}$\}} 
	Given \eqref{pri_c} and \eqref{y_to_u} with $u_{qtn}=s_{qn}c_{qtn}$ in \eqref{hybrid_c}, we use the sum-product rule to obtain 
	\begin{align}
	\label{utos}
		\Delta_{\delta {u_{qtn} \to s_{qn}}} (s_{qn}) = \mathcal{CN} \left(s_{qn} ; \frac{d_{qtn}}{p_{qtn}} , \frac{(\tau_{d_{qtn}} + \tau_{p_{qtn}})}{|p_{qtn}|^{2}} \right),
	\end{align}
	for $ n=N_{\rm P}+1,...,N$. Then the input of the decoder of $\mathbf S_{\rm D} $ is $\Delta_{ \text{DEC}_s \leftarrow s_{qn} }(s_{qn}) = \prod_{t=1}^T \Delta_{\delta {u_{qtn} \to s_{qn}}} (s_{qn})$. With $\Delta_{ \text{DEC}_s \leftarrow s_{qn} }$ $(s_{qn})$ and codebook $\mathcal C_s$, we obtain 
	the \emph{extrinsic} message of the decoder of $\mathbf S_{\rm D}$ as 
	\begin{align}
		\label{dec_s}
		\Delta_{ \text{DEC}_s \to s_{qn} }(s_{qn}) = \sum_{s \in \mathcal{S}} \alpha_{qn}(s) \delta(s_{qn} - s),
	\end{align}
	where $\alpha_{qn}(s)$ represents the probability of $s_{qn} = s$ for $s \in \mathcal S$. Combining $\Delta_{ \text{DEC}_s \to s_{qn} }(s_{qn})$ and $\Delta_{\delta {u_{qtn} \to s_{qn}}} (s_{qn})$, we obtain 
	\begin{align}
	\label{deltau_left_s}
	\Delta_{ \delta_{u_{qtn}} \leftarrow s_{qn}} (s_{qn}) & = \prod_{j \neq t} \Delta_{\delta {u_{qjn} \to s_{qn}}} (s_{qn}) \Delta_{ \text{DEC}_s \to s_{qn} }(s_{qn}) \notag \\
	& = \sum_{s \in \mathcal{S}} \pi_{qtn}(s) \delta(s_{qn} - s) 
	\end{align}
	with
	\begin{align}
		\label{pi_s}
		\pi_{qtn}(s) = \frac{\alpha_{qn}(s)   \prod_{j \neq t} \mathcal{CN}(s ; \frac{{d}_{qjn}}{{p}_{qjn}} , \frac{\tau_{d_{qjn}} + \tau_{p_{qjn}}}{|{p}_{qjn}|^{2}}) }
		{ \sum_{s \in \mathcal{S}} \alpha_{qn}(s)   \prod_{j \neq t} \mathcal{CN}(s ; \frac{{d}_{qjn}}{{p}_{qjn}} , \frac{\tau_{d_{qjn}} + \tau_{p_{qjn}}}{|{p}_{qjn}|^{2}}) }.
	\end{align}
	With \eqref{utos} and \eqref{dec_s}, the mean and the variance of $\Delta_{s_{qn}}(s_{qn})\! =\!  \prod_{t} \Delta_{\delta {u_{qtn} \to s_{qn}}} (s_{qn}) $ $\Delta_{ \text{DEC}_s \to s_{qn} }(s_{qn}),n=N_{\rm P}+1,...,N$ are
	\begin{equation}
	\begin{aligned}
		& \mu_{s_{qn}}  = \mathbb{E}[s_{qn}|\{{p}_{qtn} , {q}_{qtn}\}_{t=1}^T, s_{qn} \sim {\alpha}_{qn}; \tau_{p_{qtn}},\tau_{d_{qtn}}],   \\
		&v_{s_{qn}} = {\rm var}(s_{qn}|\{{p}_{qtn},{q}_{qtn}\}_{t=1}^T, s_{qn} \sim {\alpha}_{qn}; \tau_{p_{qtn}},\tau_{d_{qtn}}),  \label{v_s_qn}
	\end{aligned}
	\end{equation} 
	where $ p(s_{qn}|\{{p}_{qtn}, {q}_{qtn}\}_{t=1}^T, s_{qn} \! \sim \! {\alpha}_{qn}; \tau_{p_{qtn}},\tau_{d_{qtn}}) \propto \prod_t \mathcal{CN}(s_{qn} ;$ 
	$ \frac{d_{qtn}}{p_{qtn}} , \frac{\tau_{d_{qtn}} + \tau_{p_{qtn}}}{|p_{qtn}|^2}) ( \sum_{s \in \mathcal{S}} \alpha_{qn}(s) \delta(s_{qn} - s) ) $ with $s_{qn} \sim {\alpha}_{qn}$ being the abbreviation of $s_{qn} \sim ( \sum_{s \in \mathcal{S}} \alpha_{qn}(s) \delta(s_{qn} - s) )$.
	
	
	
	\subsubsection{Messages related to \{$c_{qtn}$\} and \{$x_{qtk}$\}} Similarly to \eqref{ddu2}, we have $\Delta_{\delta_{c_{qtn}} \leftarrow c_{qtn}} (c_{qtn}) = \Delta_{\delta_{u_{qtn}} \to c_{qtn}} (c_{qtn}) = \sum_{s \in \mathcal{S}} \pi_{qtn}(s) \mathcal{CN}(c_{qtn} ; d_{qtn}/s , \tau_{d_{qtn}})$ with factor node $\delta(c_{qtn}-\mathbf{f}_n^T \mathbf x_{qt})$ abbreviated as $\delta_{c_{qtn}}$. Then, we consider
	\begin{align}
		\label{D_dctox}
		\Delta_{\delta c_{qtn} \! \to x_{qtk} }\! (x_{qtk}) \! = \! \int_{c_{qtn}} \! \! \!  \Delta_{\delta_{c_{qtn}} \!  \leftarrow c_{qtn}} \! (c_{qtn}) p(c_{qtn} | x_{qtk}),
	\end{align}
	with $p(c_{qtn} | x_{qtk}) = \int_{\mathbf{x}_t/x_{qtk}} p(c_{qtn} | \mathbf x_{qt}) \prod_{j \neq k} \Delta_{x_{qtj} \to \delta c_{qtn}}$ $(x_{qtj})$. Similarly to the treatment of $p(z_{qtm} | u_{qtn})$ in \eqref{pzu2}, $p(c_{qtn}|x_{qtk})$ is treated as a Gaussian distribution. Specifically, define
	\begin{equation}
		\begin{aligned}
			\label{p_nt}
			& p_{qtn} = \sum_k f_{nk} \mu_{\delta c_{qtn} \leftarrow x_{qtk}},  \\ 
			& \tau_{p_{qtn}} = \sum_k |f_{nk}|^2 v_{\delta c_{qtn} \leftarrow x_{qtk}},
		\end{aligned}	
	\end{equation}
	where $\mu_{\delta c_{qtn} \leftarrow x_{qtk}}$ and $v_{\delta c_{qtn} \leftarrow x_{qtk}}$ are the mean and the variance of $\Delta_{x_{qtk} \to \delta c_{qtn}}(x_{qtk})$, respectively. Then, we have
	\begin{align}
		\label{p_cx}
		p( c_{qtn} | x_{qtk}) \approx \mathcal{CN}(c_{qtn}; f_{nk}(x_{qtk} - \mu_{x_{qtk}}) + p_{qtn}, \tau_{p_{qtn}} ).
	\end{align}
	Substituting \eqref{p_cx} into \eqref{D_dctox}, $\Delta_{\delta c_{qtn} \to x_{qtk} }(x_{qtk})$ is a Gaussian mixture with $|\mathcal{S}|$ components. Then $\prod_n \Delta_{\delta c_{qtn} \to x_{qtk} }(x_{qtk})$ passed to variable node $x_{qtk}$ becomes a Gaussian mixture with $|\mathcal{S}|^{N}$ components, resulting in a prohibitively high computational complexity for a relatively large $N$. Analogously to \cite[Sec. II-D]{Parker_1}, we approximate the logarithm of $\Delta_{\delta c_{qtn} \to x_{qtk} }(x_{qtk})$ in \eqref{D_dctox} by its second-order Taylor series expansion, which implies that we use a Gaussian distribution to approximate $\Delta_{\delta c_{qtn} \to x_{qtk} }(x_{qtk})$. In specific, define	 $g_{qtn}( a ) = \text{log} \int_{c_{qtn}} ( \sum_i \pi_{qtn}(s) \mathcal{CN}(c_{qtn} ;$ $ d_{qtn}/s , \tau_{d_{qtn}}) )  \mathcal{CN}(c_{qtn}; a , \tau_{p_{qtn}} )$. Then, the second-order Taylor expansion of $g_{qtn}(\cdot), n=N
	_{\rm P}+1,...,N$ at point $p_{qtn}$ yields
	\begin{align}
		\Delta_{\delta c_{qtn} \to x_{qtk} }(x_{qtk}) \approx \mathcal{CN} \Big(x_{qtk} ; & \mu_{x_{qtk}} \! +  \ \! \frac{\tau_{p_{qtn}} (\mu_{c_{qtn}}\! - \! p_{qtn}) }{f_{nk}(\tau_{p_{qtn}}\! - \! v_{c_{qtn}})}  \notag \\
		& \frac{\tau_{p_{qtn}}^2}{|f_{nk}|^2(\tau_{p_{qtn}} \!-\! v_{c_{qtn}})} \Big), \label{42}
	\end{align}
	with
	\begin{equation}
	\label{27}
	\begin{aligned}
		& {\mu}_{c_{qtn}} = \mathbb{E}[c_{qtn}|d_{qtn},p_{qtn}, s_{qn} \sim {\pi}_{qtn};\tau_{d_{qtn}},\tau_{p_{qtn}}],
		\\
		& v_{c_{qtn}} = {\rm var}(c_{qtn}|d_{qtn},p_{qtn},s_{qn} \sim {\pi}_{qtn};\tau_{d_{qtn}},\tau_{p_{qtn}}),
	\end{aligned}
	\end{equation}
	where $p(c_{qtn}|d_{qtn},p_{qtn},s_{qn} \sim {\pi}_{qtn};\tau_{d_{qtn}},\tau_{p_{qtn}}) \propto ( \sum_i \pi_{qtn}(s) $ $ \mathcal{CN}(c_{qtn} ;  d_{qtn}/s , \tau_{d_{qtn}}) ) \mathcal{CN}(c_{qtn};{p}_{qtn},\tau_{p_{qtn}})$. We obtain $\Delta_{\delta c_{qtn} \! \to x_{qtk} }(x_{qtk}), n \! = \! 1,...,N_{\rm P}$ by setting $\pi_{qtn}(s)=1$ for constellation point $s = s_{qn}$ in \eqref{42}. With \eqref{27}, the variance and mean of $\prod_{n=1}^N \Delta_{\delta c_{qtn} \to x_{qtk} }(x_{qtk})$ are  
	\begin{subequations}
	\label{30}
		\begin{align}
		& \tau_{o_{qtk}} = \Big( \sum_n |f_{nk}|^2  (\tau_{p_{qtn}} -v_{c_{qtn}})\tau_{p_{qtn}}^{-2} \Big)^{
		-1} \label{tau_o_qtk}, \\
		&{o}_{qtk} = \mu_{x_{qtk}} + \tau_{o_{qtk}}  \sum_n f_{nk}^* \tau_{p_{qtn}}^{-1}( {\mu}_{c_{qtn}} - p_{qtn} ).
	\end{align}
	\end{subequations}
	
	Similarly to the approximation of $\Delta_{y_{qtm} \to u_{qtk}}(u_{qtk})$ in \eqref{d_yu2}, we obtain $ \Delta_{y_{qtm} \to x_{qtk}}(x_{qtk}) \approx \mathcal{CN}( x_{qtk} ; \mu_{x_{qtk}} +  \frac{1}{h_{mk}} (y_{qtm}- b_{qtm}) , \frac{1}{|h_{mk}|^2}(\sigma_w^2 + \tau_{b_{qtm}}) )$. Then, we obtain $\prod_{m=1}^{M}\Delta_{y_{qtm} \to x_{qtk}} (x_{qtk}) \! = \! \mathcal{CN}(x_{qtk}; r_{qtk},\tau_{r_{qtk}})$ with 
	\begin{subequations}
		\begin{align}
		&\tau_{r_{qtk}} = \Big(\sum_m |h_{mk}|^2 (\tau_{b_{qtm}} + \sigma_w^2)^{-1} \Big)^{-1},  \label{tau_r_qtk} \\
		& {r}_{qtk} = \mu_{x_{qtk}} + \tau_{r_{qtk}} \sum_m h_{mk}^*(\tau_{b_{qtm}} + \sigma_w^2)^{-1} (y_{qtm}- {b}_{qtm}).
		\label{r_qtk}
	\end{align}
	\end{subequations}
	The inputs of the decoder of $\mathbf X$ are $ \Delta_{ \text{DEC}_x \leftarrow x_{qtk} }(x_{qtn}) = \mathcal{CN}(x_{qtk}; r_{qtk},\tau_{r_{qtk}}) \mathcal{CN}(x_{qtk}; o_{qtk},\tau_{o_{qtk}}),\forall q,t,k$. With $\mathcal C_x$ and $\Delta_{ \text{DEC}_x \leftarrow x_{qtk} }(x_{qtn})$, the \emph{extrinsic} message of the decoder of $\mathbf X$ is
	\begin{align}
		\label{dec_x}
		\Delta_{ \text{DEC}_x \to x_{qtk} }(x_{qtk}) = \sum_{x \in \mathcal{X}} \beta_{qtk}(x) \delta(x_{qtk} - x).
	\end{align}
	Then, the mean and the variance of $\Delta_{x_{qtk}}(x_{qtk})=\Delta_{ \text{DEC}_x \leftarrow x_{qtk} }(x_{qtk})$ $\Delta_{ \text{DEC}_x \to x_{qtk} }(x_{qtk})$ are 
	\begin{equation}
	\begin{aligned}
		\mu_{x_{qtk}} &= \mathbb{E}[x_{qtk}|r_{qtk},o_{qtk}, x_{qtk}\sim {\beta}_{qtk}; \tau_{r_{qtk}},\tau_{o_{qtk}}], 
		\\
		v_{x_{qtk}} & = {\rm var}(x_{qtk}|r_{qtk},o_{qtk}, x_{qtk}\sim {\beta}_{qtk}; \tau_{r_{qtk}},\tau_{o_{qtk}}),
		\label{v_x_qtk}  
	\end{aligned} 
	\end{equation}
	where $p(x_{qtk}|r_{qtk},o_{qtk} , x_{qtk} \! \sim \! {\beta}_{qtk};\tau_{r_{qtk}} ,\tau_{o_{qtk}}) \! \propto \! \mathcal{CN}(x_{qtk} ;$  $r_{qtk} , \tau_{r_{qtk}} ) \mathcal{CN}(x_{qtk} ; o_{qtk} , \tau_{o_{qtk}} ) ( \sum_{x \in \mathcal{X}} \beta_{qtk}(x) \delta(x_{qtk}- x) ) $.
	
	To reduce complexity, we approximate $\mu_{y_{qtm} \leftarrow x_{qtk}}$ and $\mu_{c_{qtn} \leftarrow x_{qtk}}$ by following \eqref{u_ext} as
	\begin{align}
		& \mu_{y_{qtm} \leftarrow x_{qtk}}  \approx \mu_{x_{qtk}} - v_{x_{qtk}}  h_{mk}^* (\sigma_w^2+\tau_{b_{qtm}})^{-1}  (y_{qtm}-b_{qtm}), \label{x_to_y}
			\\
		& \mu_{c_{qtn} \leftarrow x_{qtk}} \approx \mu_{x_{qtk}} - v_{x_{qtk}} f_{nk}^* (\tau_{p_{qtn}})^{-1} (\mu_{c_{qtn}}-p_{qtn}) \label{x_to_c}.
	\end{align}
	Substituting \eqref{x_to_c} and $v_{x_{qtk}} \approx v_{\delta c_{qtn} \leftarrow x_{qtk}}$ into \eqref{p_nt}, we obtain
	\begin{equation}
		\label{p_aprox}
		\begin{aligned}
			& p_{qtn} \approx  \sum_k f_{nk} \mu_{x_{qtk}} - \tau_{p_{qtn}}(\tau_{p_{qtn}}^{\rm p})^{-1}({\mu}_{c_{qtn}} - p_{qtn}^{\rm p}), \\
			& \tau_{p_{qtn}} \approx \sum_k |f_{nk}|^2 v_{x_{qtk}},
		\end{aligned}	
	\end{equation}
	where superscript ``$\rm p$" is used to indicate that $p_{qtn}^{\rm p}$ and $\tau_{p_{qtn}}^{\rm p}$ are obtained from the previous iteration. We then substitute \eqref{u_ext} and \eqref{x_to_y} into \eqref{btaub} to obtain
	\begin{subequations}
	\begin{align}
		& b_{qtm} \approx \sum_{ n} g_{mn} \mu_{ u_{qtn} } + \sum_k h_{mk} \mu_{x_{qtk}}  \notag \\
		& \hspace{1.2cm} - \tau_{b_{qtm}}(\sigma_w^2+\tau_{b_{qtm}}^{\rm p})^{-1} (y_{qtm}-b_{qtm}^{\rm p}), \label{b_qtm} \\
		& \tau_{b_{qtm}}\approx \sum_{n} |g_{mn}|^2 v_{u_{qtn}} + \sum_k |h_{mk}|^2 v_{x_{qtk}}. \label{tau_b_qtm}
	\end{align}	
	\end{subequations}
	
	\vspace{-0.2cm}
	\subsubsection{Scalar variances}
	 We further show that the variances of messages can be approximated by scalers. We first define
	\begin{align}
		&v_u = \frac{1}{QTN} \sum_{q,t,n} v_{u_{qtn}}, ~ v_s = \frac{1}{QN} \sum_{q,n} v_{s_{qn}}, \notag \\
		& v_c = \frac{1}{QTN} \sum_{q,t,n} v_{c_{qtn}}, ~ {\rm and} ~ v_x = \frac{1}{QTK} \sum_{q,t,k} v_{x_{qtk}}, \label{v_four}
	\end{align}
	in place of $v_{u_{qtn}}$, $v_{s_{qn}}$, $v_{c_{qtn}}$, and $v_{x_{qtk}}$ respectively.
	
	We next approximate $\tau_{b_{qtm}}$ in \eqref{tau_b_qtm}. Note that message passings of $x_{qtk}$ at different sub-block $q$ and time-slot $t$  are identical, leading to $\frac{1}{K} \sum_k v_{x_{qtk}} \approx v_x, \forall q,t$. Similarly, we obtain $\frac{1}{N} \sum_n v_{u_{qtn}} \approx v_u, \forall q,t$. The powers of rows of $\mathbf G$ (or $\mathbf H$) commonly tend to be equal as matrix size increases, i.e., $\sum_n |g_{mn}|^2 \approx \frac{1}{M} \| \mathbf G\|_F^2, \forall m $ (or $\sum_k |h_{mk}|^2 \approx \frac{1}{M} \| \mathbf H\|_F^2, \forall m$),  known as \emph{channel hardening} \cite{tse2005fundamentals}. Substituting these approximations into \eqref{tau_b_qtm}, we obtain 
	\begin{align}
		\tau_{b_{qtm}} \approx \frac{1}{M} \| \mathbf G\|_F^2 v_u +  \frac{1}{M} \| \mathbf H\|_F^2 v_x  \triangleq \tau_b. \label{apr_b}
	\end{align}
	Similarly to \eqref{apr_b}, $\tau_{d_{qtn}}$, $\tau_{o_{qtk}}$, $\tau_{r_{qtk}}$, and $\tau_{p_{qtn}}$ in \eqref{tau_qnt}, \eqref{tau_o_qtk}, \eqref{tau_r_qtk}, and \eqref{p_aprox}  are approximated as
	\begin{align}
		 & \tau_{d_{qtn}} \! \approx \! \frac{N(\sigma_w^2+\tau_b)}{\| \mathbf G\|_F^2}  \triangleq \tau_d, ~
		\tau_{o_{qtk}} \! \approx \! \frac{K \tau_p^{2}}{\|\mathbf F\|_F^2 ( \tau_p - v_c )} \triangleq \tau_o, \notag \\
		& \tau_{r_{qtk}} \! \approx \! \frac{K(\sigma_w^2 + \tau_b)}{\|\mathbf H\|_F^2}  \triangleq \tau_r, ~ {\rm and} ~ 
		 \tau_{p_{qtn}} \! \approx \! \frac{\| \mathbf F \|_F^2 v_x}{N}  \triangleq \tau_p. \label{apr}
	\end{align}

	\vspace{-0.2cm}
	\subsection{Overall Algorithm}

	The overall algorithm, as summarized in Algorithm 1, involves mean and variance computations of messages. To start with, module A outputs mean-variance pair $(d_{qtn},\tau_d)$ of $\prod_m \Delta_{y_{qtm} \to u_{qtn}}(u_{qtn})$ in step 1. Then, module B outputs $(p_{qtn},\tau_p)$ of $\Delta_{\delta_{c_{qtn}} \to c_{qtn}}(c_{qtn})$ in step 2, which is also the input of Module C. Combining $(d_{qtn},\tau_d)$ and $(p_{qtn},\tau_p)$ together with $\pi_{qtn}(s)$ (in step 9), we update the mean-variance pairs of $\Delta_{u_{qtn}}(u_{qtn})$ and $\Delta_{c_{qtn}}(c_{qtn})$ in steps 3 and 4, respectively. In steps 5-6, we update the mean-variance pairs of $\prod_m \Delta_{y_{qtm} \to x_{qtk}}(x_{qtk})$ and $ \prod_n \Delta_{\delta_{c_{qtn}} \to x_{qtk}}$ $(x_{qtk})$. Combining these two messages together with $\beta_{qtk}(x)$ in step 7, we update the mean-variance pair of $\Delta_{x_{qtk}}(x_{qtk})$ in step 8. Similarly, by using $(d_{qtn},\tau_d)$ in step 1, $(p_{qtn},\tau_p)$ in step 2, and $\alpha_{qn}(s)$ in step 9, we update the mean-variance pairs of $\Delta_{s_{qn}}(s_{qn})$ in step 10. Step 11 updates the auxiliary variables for the subsequent iterations. Algorithm 1 is executed until convergence, where $\mu_{x_{qtk}}$ and $\mu_{s_{qn}}$ are the estimates of $x_{qtk}$ and $s_{qn}$, respectively.

	Algorithm 1 can be applied to some special cases with minor modifications. For example, for an uncoded system, we fix $\alpha_{qn}(s) = 1 / |\mathcal{S}|, n=N_{\rm P}+1,...,N$ and $\beta_{qtk}(x) = 1 / |\mathcal{X}|$ in steps 7 and 9. For a coded system with separate detection and decoding, we can fix $\beta_{qtk}(x) $ and $\alpha_{qn}(s)$ during the iterative process, and only update them at the final iteration. For the case of a blocked direct link, i.e., $\mathbf H = \mathbf 0$, we delete the messages between variable node $x_{qtk}$ and factor node $p(y_{qtm}|\mathbf u_{qt}, \mathbf x_{qt})$. Specifically, we delete the computation of $\tau_r$ and $r_{qtk}$ in step 5, and replace steps 8 and 11 with $ \mu_{x_{qtk}} = \mathbb{E}[x_{qtk}|o_{qtk}, x_{qtk}\sim {\beta}_{qtk};\tau_{o}] $, $ v_x = \frac{1}{KTQ} \sum_{k,t,q} {\rm var}(x_{qtk}|o_{qtk}, x_{qtk}\sim {\beta}_{qtk};\tau_{o}) $, $\tau_b=\frac{1}{M} \| \mathbf G\|_F^2 v_u$, and $b_{qtm} = \sum_{ n} g_{mn} \mu_{ u_{qtn} } - \tau_b (\sigma_w^2+\tau_b^{\rm p})^{-1} (y_{qtm}-b_{qtm}^{\rm p}) $, respectively.
	
	The algorithm complexity is dominated by the matrix and vector multiplications in steps 1, 2, 5, 6, and 11 with complexity $\mathcal{O}\big(QT(KN+KM+NM+NK) \big)$.
	Steps 3-4, 8, and 10 only involve point-wise operations. The decoding complexity in steps 7 and 9 is determined by the specific codes employed in the system. For commonly used codes such as convolutional codes and low-density parity-check (LDPC) codes, the decoding complexity is linear to the code length. To summarize, the algorithm complexity is linear to the system parameters $N$, $M$, $K$, $Q$, and $T$.
	
	{We now compare the complexity of our SAPIT method with the methods in \cite{Wenjing_JSAC} and \cite{Shuaishuai}. Assume that the methods in \cite{Wenjing_JSAC} and \cite{Shuaishuai} have the same transmission rate as our SAPIT method, i.e., $K \log_2 |\mathcal X|$ and $(N-N_{\rm p}) \log_2|\mathcal S|$ bits (per channel use) at the Tx and the RIS, respectively. Besides, we ignore the complexity of passive beamforming design in \cite{Wenjing_JSAC}. The main computational computation in \cite{Wenjing_JSAC} is from the TMP algorithm, which is $\mathcal{O}(QTKNM + QTM(N+K))$ in each TMP iteration. The computational complexity in \cite{Shuaishuai} is mainly from two parts, i.e., the reflection pattern optimization and the maximum likelihood (ML) detection. The complexity involved in pattern optimization is $\mathcal{O}\big( (N-N_{\rm p})^3 (\log_2|\mathcal S|)^3 K^3 (\log_2 |\mathcal X|)^3 N^2(M+K)+ (N-N_{\rm p})^4 (\log_2|\mathcal S|)^4 K^4 (\log_2 |\mathcal X|)^4 M \big)$. The complexity of the ML detection is roughly $2^{K \log_2 |\mathcal X| + (N-N_{\rm p}) \log_2|\mathcal S|}$. Therefore, the computational complexities involved in \cite{Wenjing_JSAC} and \cite{Shuaishuai} are both much higher than that of our proposed SAPIT method.
	}
	
	\begin{algorithm}[h]
		\small
	   \caption{\label{alg1}The Proposed Algorithm}
	   \begin{algorithmic}
		
		   \REQUIRE $\mathbf{Y}$, $\mathbf G$, $\mathbf F$, $\mathbf H$, $\sigma_w^2$
		   
		   \ Initialization: $\mu_{x_{qtk}} \! = \!\mu_{c_{qtn}} \! = \! \mu_{u_{qtn}} \! = \! p_{qtn} \! = \! b_{qtm} \! = \! 0$, $v_x \! = \! 1$,
		   
		   ~  $\tau_p \! = \! v_u \! = \! v_c \! = \frac{\| \mathbf F \|_F^2}{N}$, $\tau_b \! = \! \frac{\| \mathbf G\|_F^2}{M}  v_u  \! + \!  \frac{\| \mathbf H\|_F^2}{M} v_x$, $\pi_{qtn}(s)\! = \! \frac{1}{|\mathcal S|}$
		   
		   \	\textbf{while} the stopping criterion is not met  \textbf{do}
	
		   ~~~ \% Mean-variance pairs from module A to nodes $\{u_{qtn}\}$
		   
		   \quad \ 1: $ \tau_d =  N  (\| \mathbf G\|_F^2)^{-1}   (\sigma_w^2+\tau_b)$ 
		   
		   \quad ~~~ $ d_{qtn} = \mu_{ u_{qtn} } + N  (\| \mathbf G\|_F^2)^{-1}  \sum_{m} g_{mn}^* (y_{qtm}-b_{qtm}) $	
		   \vspace{0.05cm}
	
		   ~~~ \% Mean-variance pairs from module B to nodes $\{c_{qtn} \}$
	
		   \quad \ 2: $\tau_p  =  \frac{1}{N} \| \mathbf F \|_F^2 v_x $ 
	
		   \quad ~~~ $p_{qtn} =   \sum_k f_{nk} \mu_{x_{qtk}} - \tau_{p}(\tau_{p}^{\rm p})^{-1}({\mu}_{c_{qtn}} - p_{qtn}^{\rm p}) $
	
		   \vspace{0.05cm}
		   
		   ~~~ \% Mean-variance pairs at variable nodes $\{u_{qtn} \}$
	
		   \quad \ 3: $\mu_{ u_{qtn} } \! = \! \mathbb{E}[u_{qtn}|d_{qtn},p_{qtn},s_{qn}\sim {\pi}_{qtn};\tau_{d},\tau_{p}] $ 
		   
		   \quad ~~~ $ v_u \! = \! \frac{1}{NTQ} \sum_{n,t,q} {\rm var}(u_{qtn}|d_{qtn},p_{qtn},s_{qn}\sim {\pi}_{qtn};\tau_{d},\tau_{p})$
	   
		   \vspace{0.05cm}
	
		   ~~~ \% Mean-variance pairs at variable nodes $\{c_{qtn} \}$
	
		   \quad \ 4: $ \mu_{ c_{qtn} } = \mathbb{E}[c_{qtn}|d_{qtn},p_{qtn},s_{qn}\sim {\pi}_{qtn};\tau_{d},\tau_{p}] $ 
		   
		   \quad ~~~ $ v_c = \frac{1}{NTQ} \sum_{n,t,q} {\rm var}(c_{qtn}|d_{qtn},p_{qtn},s_{qn}\sim {\pi}_{qtn};\tau_{d},\tau_{p})$
	
		   \vspace{0.05cm}
	
		   ~~~ \% Mean-variance pairs from modules A and B to nodes $\{x_{qtk} \}$
	
		   \quad \ 5: $\tau_r = K (\|\mathbf H\|_F^2)^{-1} (\sigma_w^2 + \tau_b)$ 
		   
		   \quad ~~~ ${r}_{qtk} = \mu_{x_{qtk}} + K (\|\mathbf H\|_F^2)^{-1} \sum_m h_{mk}^* (y_{qtm}-{b}_{qtm})$
	
		   \vspace{0.05cm}
	
		   \quad \ 6: $\tau_o = K (\|\mathbf F\|_F^2)^{-1} \tau_p^{2}/( \tau_p - v_c )$ 
		   
		   \quad ~~~ ${o}_{qtk} = \mu_{x_{qtk}} + K (\|\mathbf F\|_F^2)^{-1} \frac{\tau_p}{\tau_p -v_c}\sum_n f_{nk}^*( {\mu}_{c_{qtn}} \! - \! p_{qtn} ) $
	
		   \vspace{0.05cm}
	
		   \quad \ 7: Update $\beta_{qtk}(x)$ in (30) from $\mathbf X$-decoder
		   
		   \vspace{0.05cm}
	
		   ~~~ \% Mean-variance pairs at variable nodes $\{ x_{qtk} \}$
	
		   \quad \ 8: $ \mu_{x_{qtk}}  =  \mathbb{E}[x_{qtk}|r_{qtk},o_{qtk}, x_{qtk}\sim {\beta}_{qtk};\tau_{r},\tau_{o}] $ 
		   
		   \quad ~~~ $ v_x  =  \frac{1}{KTQ} \sum_{k,t,q} {\rm var}(x_{qtk}|r_{qtk},o_{qtk}, x_{qtk}\sim {\beta}_{qtk};\tau_{r},\tau_{o}) $ 
	
		   \vspace{0.05cm}
	
		   \quad \ 9: Update $\alpha_{qn}(s)$ in (19) and $\pi_{qtn}(s)$ in (21) from $\mathbf S_{\rm D}$-decoder
	
		   ~~~ \% Mean-variance pairs at variable nodes $\{ s_{qn} \}$
	
		   \quad \  10: $\mu_{s_{qn}} \! = \! \mathbb{E}[s_{qn}|\{{p}_{qtn} , d_{qtn}\}_t, s_{qn} \! \sim \! {\alpha}_{qn}; \tau_{p},\tau_{d}] $ 
		   
		   \quad ~~~~~ $ v_s \! = \! \frac{1}{NQ} \sum_{n,q} {\rm var}(s_{qn}|\{{p}_{qtn} , d_{qtn}\}_t, s_{qn} \! \sim \! {\alpha}_{qn}; \tau_{p},\tau_{d})$ 
		   
		   \vspace{0.05cm}
	
		   ~~~ \% Quantities for subsequent iterations
	
		   \quad \ 11: $\tau_b = \frac{1}{M} \| \mathbf G\|_F^2 v_u +  \frac{1}{M} \| \mathbf H\|_F^2 v_x$ 
		   
		   \quad ~~~~~ $b_{qtm} = \sum_{ n} g_{mn} \mu_{ u_{qtn} } \! + \! \sum_k h_{mk} \mu_{x_{qtk}} - \frac{\tau_b (y_{qtm}-b_{qtm}^{\rm p})}{\sigma_w^2+\tau_b^{\rm p}} $
	
		   \textbf{end }
		   
		   \ENSURE $\mu_{x_{qtk}} $, $ \mu_{s_{qn}}$
		
	   \end{algorithmic}
	\end{algorithm}
	
	\section{State Evolution Analysis}
	In this section, we characterize the behavior of Algorithm 1 in the large-system limit. We first present a heuristic description of the SE in Sec. IV-A. We then prove the SE in Sec. IV-B by borrowing the techniques developed in \cite{sundeep} and \cite{AMP_SE}. Based on the SE, we analyze the achievable rate region of the SAPIT-MIMO system in Sec. IV-C.
	
	\subsection{SE Description}

	\begin{figure*}[h] 
		\centering
		\includegraphics[width = 6 in]{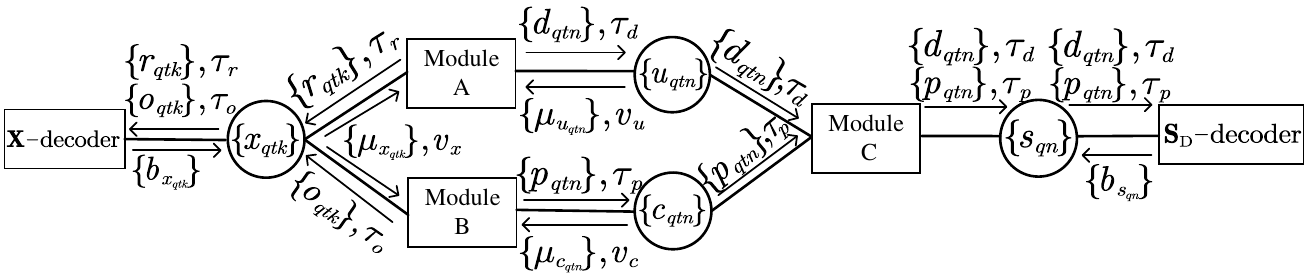}
		\vspace{-0.1cm}
		\caption{{A simplified version of the factor graph in Fig. 2. Modules A, B, and C represent factor nodes $\{p(y_{qtm}|\mathbf u_{qt},\mathbf x_{qt})\}$, $\{\delta(c_{qtn}-\mathbf f_n^T \mathbf x_{qt})\}$, and $\{ \delta(u_{qtn}-s_{qn}c_{qtn})\}$, respectively; the four circles are super variable nodes of $\{x_{qtk}\}$, $\{u_{qtn}\}$, $\{c_{qtn}\}$, and $\{s_{qn}\}$; the variables above (or below) an arrow represent the input/output of the corresponding module.}}
		\label{App}
	\end{figure*}

	For convenience of analysis, we assume $||\mathbf G||_F^2/N =1$, $||\mathbf H||_F^2/K =1$, and $||\mathbf F||_F^2/K = \zeta$. This assumption does not lose generality. To see this, we construct an equivalent system model $\frac{\mathbf{Y}_q}{a} = \left( \frac{ \mathbf G}{b} \text{diag}(\mathbf s_q) \frac{b \mathbf F}{a} + \frac{\mathbf H}{a} \right) \mathbf X_{q} + \frac{\mathbf{W}_q}{a}$ with $a=\sqrt{||\mathbf H||_F^2/K},b=\sqrt{||\mathbf G||_F^2/N}$. Then we define $\tilde{\mathbf G} = \frac{ \mathbf G}{b}$, $\tilde{\mathbf H} = \frac{ \mathbf H}{a}$, and $\tilde{\mathbf F} = \frac{b \mathbf F}{a}$ as the equivalent channel matrices. It is easy to verify that $\tilde{\mathbf G}$, $\tilde{\mathbf H}$, and $\tilde{\mathbf F}$ satisfies the above power assumption by setting $\zeta = \frac{||\mathbf G||_F^2 ||\mathbf F||_F^2}{N ||\mathbf H||_F^2}$. 
	
	We regard the means of messages as the estimates of the corresponding random variables. We show in the next subsection that the message variances $(v_x,v_u,v_c,v_s,\tau_d,\tau_r,\tau_p,\tau_o)$ converge to the MSEs of the corresponding random variables in the large-system limit. For the convenience of discussion here, we use $(v_x,v_u,v_c,v_s,\tau_d,\tau_r,\tau_p,\tau_o)$ to represent the corresponding MSEs with some abuse of notation. We refer to $(v_x,v_u,v_c,v_s,\tau_d,\tau_r,\tau_p,\tau_o)$ as the state variables in the SE. For illustration, Fig. \ref{App} is a simplified version of Fig. \ref{FG}.
	The following subsections show how the state variables are associated at each module or super variable node.
	
	
	 \subsubsection{Transfer function of module A}
	 Module A infers $\{x_{qtk}\}$ and $\{u_{qtn}\}$ based on \eqref{hybrid_a}, where the input estimate and the corresponding MSE of $x_{qtk}$ (or $u_{qtn}$) are respectively $\mu_{ x_{qtk}}$ (or $\mu_{ u_{qtn}}$) and $v_x$ (or $v_u$), and the output estimate and the corresponding MSE of $x_{qtk}$ (or $u_{qtn}$) are respectively $r_{qtk}$ (or $d_{qtn}$) and $\tau_r$ (or $\tau_d$). Out goal is to describe how $\tau_r$ and $\tau_d$ vary as a function of $v_x$ and $v_u$, termed as the transfer function of module A. A key observation is that due to the linear mixing in \eqref{hybrid_a}, the output means $\{r_{qtk}\}$ (or $\{d_{qtn}\}$) can be well approximated by random samples over an AWGN channel. This approximation can be made rigorous in the large-system limit, as detailed in the next subsection. To be specific, we can model the outputs of module A as 
	\begin{align}
		& r_{qtk} = x_{qtk} + w_{r_{qtk}}, \label{R1} \\
		& d_{qtn} = u_{qtn} + w_{d_{qtn}}, \label{D}
	\end{align}
	where $\{w_{r_{qtk}}\}$ are i.i.d. from $\mathcal{CN}(0,\tau_{r})$, $\{x_{qtk}\}$ are i.i.d. with $x_{qtk} \sim \sum_{x \in \mathcal X} \frac{1}{|\mathcal{X}|} \delta(x_{qtk}-x)$, and $w_{r_{qtk}}$ is independent of $x_{qtk}$; $\{w_{d_{qtn}}\}$ are i.i.d. from $\mathcal{CN}(0,\tau_{d})$, $\{u_{qtn}\}$ are i.i.d. with $u_{qtn} \sim \sum_{s \in \mathcal S} \frac{1}{|\mathcal{S}|} \mathcal{CN}(u_{qtn};0,\zeta \frac{K}{N})$, and $w_{d_{qtn}}$ is independent of $u_{qtn}$.
	Substituting $||\mathbf G||_F^2/N = ||\mathbf H||_F^2/K =1$ and \eqref{apr_b} into \eqref{apr}, we obtain the transfer function of module A as
	\begin{align}
	\tau_d = \tau_{r} = \frac{N}{M}v_u+\frac{K}{M}v_x+\sigma_w^2. 
	\label{tau_r}
	\end{align}
	We next give a heuristic explanation of $r_{qtk}$ in \eqref{R1} and $\tau_r$ in \eqref{tau_r}. By substituting \eqref{hybrid_a}, \eqref{b_qtm},  \eqref{apr_b}, and \eqref{apr} into \eqref{r_qtk}, we obtain $r_{qtk} = x_{qtk} + w_{r_{qtk}}$ where $w_{r_{qtk}} = \sum_{m,j\neq k} h_{mk}^* h_{mj} (x_{qtj}-\mu_{ x_{qtj}}) + \sum_{m,n} h_{mk}^* g_{mn} (u_{qtn}-\mu_{ u_{qtn}}) + \sum_m h_{mk}^* w_{qtm} + \Delta_{qtk}$ with $\Delta_{qtk} = \sum_m h_{mk}^* \tau_b(\sigma_w^2+\tau_b^{\rm p})^{-1}(y_{qtm}-b_{qtm}^{\rm p})$.  $\Delta_{qtk}$ is an ``Onsager reaction term" \cite{AMP_2009} to ensure the independence between $w_{r_{qtk}}$ and $x_{qtk}$. From the randomness of $\{h_{mk},x_{qtj}-\mu_{x_{qtj}},g_{mn},u_{qtn} - \mu_{ u_{qtn}}\}$, $w_{r_{qtk}}$ are approximately Gaussian in a system of a large size. Note that $\mu_{u_{qtn}}$ and $\mu_{x_{qtk}}$ are estimates of $u_{qtn}$ and $x_{qtk}$ with MSEs $v_u$ and $v_x$, respectively. Then, we obtain $\tau_r = \frac{K-1}{M}v_x + \frac{N}{M} v_u + \sigma_w^2 \approx \frac{K}{M}v_x + \frac{N}{M} v_u + \sigma_w^2$.\footnote{Here we ignore the contribution of the Onsager reaction term $\Delta_{qtk}$, and leave the rigorous analysis to Subsection B.} The above discussions can be applied to $d_{qtn}$ in \eqref{D} and $\tau_d$ in \eqref{tau_r} similarly.

	\subsubsection{Transfer function of module B}
	Module B infers $\{x_{qtk}\}$ and $\{c_{qtn}\}$ based on \eqref{hybrid_b}, where the input estimate and the corresponding MSE of $x_{qtk}$ (or $c_{qtn}$) are respectively $\mu_{x_{qtk}}$ (or $\mu_{c_{qtn}}$) and $v_x$ (or $v_u$), and the output estimate and the corresponding MSE of $x_{qtk}$ (or $c_{qtn}$) are respectively $o_{qtk}$ (or $p_{qtn}$) and $\tau_o$ (or $\tau_p$). We model the output of module B by
	\begin{align}
		& o_{qtk} = x_{qtk} + w_{o_{qtk}},   \label{R2} \\
		& c_{qtn} = p_{qtn} + w_{c_{qtn}},	\label{C}
	\end{align}
	where $\{w_{o_{qtk}}\}$ are i.i.d. from $\mathcal{CN}(0,\tau_{o})$ and $w_{o_{qtk}}$ is independent of $x_{qtk}$; $\{p_{qtn}\}$ are i.i.d. from $\mathcal{CN}(0,\zeta \frac{K}{N} - {\tau}_{p})$, $\{w_{c_{qtn}}\}$ are i.i.d. from $\mathcal{CN}(0, \tau_p)$, and $p_{qtn}$ is independent of $w_{c_{qtm}}$. 
	Substituting $||\mathbf F||_F^2/K = \zeta$ into \eqref{apr}, the transfer function of module B is
	\begin{align}
	\quad \tau_{o} = \frac{{\tau}_p^2}{\zeta({\tau}_p -{v}_c)} \quad { \rm and } \quad \tau_p = \zeta \frac{K}{N} v_x \label{tau_po}.
	\end{align}
	The reasoning behind $o_{qtk}$ in \eqref{R2} and $\tau_o$ in \eqref{tau_po} are similar to that of $r_{qtk}$ in \eqref{R1} and $\tau_r$ in \eqref{tau_r}, respectively. We now give an explanation of $c_{qtn}$ in \eqref{C} and $\tau_p$ in \eqref{tau_po}. From the form of $p_{qtn}$ in \eqref{p_aprox} and $\mathbf C_q \! = \! \mathbf F \mathbf X_q $ in \eqref{hybrid_b}, we obtain $c_{qtn} \! = \! p_{qtn}  \! + \! w_{c_{qtn}}$ where $p_{qtn}\! = \! \sum_k f_{nk} \mu_{x_{qtk}} \! + \! \tilde{\Delta}_{qtn}$ and $w_{c_{qtn}} \! = \! \sum_k f_{nk} (x_{qtk} \! - \! \mu_{x_{qtk}}) \! - \! \tilde{\Delta}_{qtn}$ with $\tilde{\Delta}_{qtn} \! = \! \tau_{p_{qtn}}(\tau_{p_{qtn}}^{\rm p})^{-1}({\mu}_{c_{qtn}} \!-  p_{qtn}^{\rm p})$. Note that $\tilde{\Delta}_{qtn}$ is similar to the Onsager term $\Delta_{qtk}$ above, and is ignored here for a heuristic discussion. From the randomness of $\{f_{nk},x_{qtk},\mu_{x_{qtk}} \}$, $p_{qtn}$, and $w_{c_{qtn}}$ are approximated as Gaussian random variables, of which the independence is equivalent to uncorrelation. Given that $\{x_{qtk}\}$ (or $\{ \mu_{x_{qtk}} \}$) are i.i.d., the uncorrelation of $p_{qtn}$ and $w_{c_{qtn}}$ is ensured if $\mu_{x_{qtk}}$ and $x_{qtk}-\mu_{x_{qtk}}$ are uncorrelated. From the discussions later in \eqref{v_x}, we show that $\mu_{x_{qtk}}$ can be expressed as the conditional expectation of $x_{qtk}$ given $r_{qtk}$ in \eqref{R1}, $o_{qtk}$ in \eqref{R2}, and $b_{x_{qtk}}$ from $\mathbf X$-decoder, i.e., $\mu_{x_{qtk}} \! = \! \mathbb{E} [x_{qtk}|r_{qtk},o_{qtk}, b_{x_{qtk}} ]$. From the orthogonal principle \cite{Kay}, $\mu_{x_{qtk}}$  and $x_{qtk}-\mu_{x_{qtk}}$ are uncorrelated. $\tau_p$ in \eqref{tau_po} can be obtained from $w_{c_{qtn}} = \sum_k f_{nk} (x_{qtk}-\mu_{x_{qtk}})$ together with the fact that $v_x$ is the MSE of $\mu_{{x}_{qtk}}$.
	
	\subsubsection{Transfer function of the decoder of $\mathbf X$}
	We need to characterize the $\mathbf X$-decoder soft outputs $\{ \Delta_{ \text{DEC}_x \to x_{qtk} }(x_{qtk}) = \sum_{x \in \mathcal{X}} \beta_{qtk}(x) \delta(x_{qtk} - x) \}$, where the inputs of the decoder are the AWGN observations $\{r_{qtk}\}$ in \eqref{R1} and $\{o_{qtk}\}$ in \eqref{R2} with noise variances $\tau_r$ and $\tau_o$, respectively. To this end, we introduce random variables $\{b_{x_{qtk}}\}$ as the auxiliary outputs that carry information from the decoder of $\mathbf X$. That is, given $\{b_{x_{qtk}}\}$, the decoder soft outputs can be expressed as $\Delta_{ \text{DEC}_x \to x_{qtk} }(x_{qtk}) = p(x_{qtk}|b_{x_{qtk}}, x_{qtk} \in \mathcal{X}) = \sum_{x \in \mathcal{X}} \beta_{qtk}(x) \delta(x_{qtk} - x)$, where $\beta_{qtk}(x)$ is a function of $b_{x_{qtk}}$. Note that the forms of $\{b_{x_{qtk}}\}$ rely on specific coding schemes, e.g., $b_{x_{qtk}}$ is a log-likelihood ratio when binary phase-shift keying (BPSK) modulation is applied. We can obtain $\{b_{x_{qtk}}\}$ by simulating the $\mathbf X$-decoder over the AWGN channels in \eqref{R1} and \eqref{R2}.
	
	\subsubsection{State equation at super variable node $\{x_{qtk}\}$}
	The operation at super variable node $\{x_{qtk}\}$ is to generate an output estimate $\mu_{x_{qtk}}$ for each $x_{qtk}$ and the corresponding MSE $v_{x}$ by combining the AWGN observations $r_{qtk}$ and $o_{qtk}$, and the decoder soft output $p(x_{qtk}|b_{x_{qtk}} , x \in \mathcal{X})$. From $\{r_{qtk}\}$ in \eqref{R1}, $\{o_{qtk}\}$ in \eqref{R2}, and $\{b_{x_{qtk}}\}$ mentioned above, we have  
	\begin{align}
	 {v}_x = \mathbb{E} [ {\rm var}(x|r,o, x \sim p(x|b_{x}, x \in \mathcal{X} ) ; \tau_{r},{\tau}_{o}) ], \label{v_x}
	\end{align}
	where $\mathbb{E}[\cdot]$ is taken over $r_{qtk}$, $o_{qtk}$, and $b_{x_{qtk}}$. Note that the subscripts $q,t,k$ of $\{x_{qtk}, r_{qtk}, o_{qtk}, b_{x_{qtk}}\}$ are omitted in \eqref{v_x} since $v_x$ is invariant to these subscripts.
	
	
	\subsubsection{State equations related to module C} 
	
	What remain are the characterizations of super nodes $\{s_{qn}\}$, $\{u_{qtn}\}$, and $\{c_{qtn}\}$, decoder of $\mathbf S_{\rm D}$, and module C. We start from the outputs of module C for $\{s_{qn}\}$. Note that $\{s_{qn}\}_{n=1}^{N_{\rm P}}$ are known pilots. Thus, we only consider the model of the outputs $\{s_{qn}\}_{n=N_{\rm P}+1}^{N}$. 
	Specifically, substituting \eqref{hybrid_c} and \eqref{C} into \eqref{D}, we obtain
	\begin{align}
		d_{qtn} = p_{qtn} s_{qn} + w_{dc_{qtn}}, ~ t = 1,...,T, \label{S}
	\end{align}
	where $s_{qn} \! \sim \! \sum_{s \in \mathcal S} \frac{1}{|\mathcal{S}|} \delta(s_{qn}-s)$; $w_{dc_{qtn}} = s_{qn} w_{c_{qtn}} + w_{d_{qtn}} \sim \mathcal{CN}(0,\tau_d + \tau_p)$ is independent of $s_{qn}$. 
	 Note that \eqref{S} is a Rayleigh fading model of $s_{qn}$.
	
	 We next characterize the soft outputs of the $\mathbf S_{\rm D}$-decoder. Similarly to the $\mathbf X$-decoder, we introduce $\{b_{s_{qn}}\}$ as the auxiliary outputs. Then the soft outputs are expressed as $\Delta_{ \text{DEC}_s \to s_{qn} }(s_{qn}) = p(s_{qn}|b_{s_{qn}},s_{qn}\in \mathcal{S})=\sum_{s \in \mathcal S} \alpha_{qn}(s) \delta(s_{qn}-s) $ with $\alpha_{qn}(s)$ treated as a function of $b_{s_{qn}}$. Combining \eqref{S} and $p(s_{qn}|b_{s_{qn}},s_{qn}\in \mathcal{S})$, we obtain the state equation at super variable node $\{s_{qn}\}$ as
	\begin{align}
	 {v}_s = \mathbb{E} [ {\rm var}(s|\{{p}_{t} , d_{t}\}_{t=1}^T, s \sim p(s| b_{s},s\in \mathcal{S}); \tau_{p},\tau_{d}) ], \label{v_s} 
	\end{align}
	where $\mathbb{E}[\cdot]$ is taken over $\{p_{qtk},d_{qtk}\}_{t=1}^T$ and $b_{s_{qn}}$.
	
	Consider the soft inputs of module C from super variable node $\{s_{qn}\}$. Using \eqref{S} and $p(s_{qn}|b_{s_{qn}},s_{qn}\in \mathcal{S})$, we express $\{\Delta_{\delta_{u_{qtn}} \leftarrow s_{qn}}(s_{qn})\}_{n=N_{\rm P}+1}^N$ in \eqref{deltau_left_s} as $ \Delta_{\delta_{u_{qtn}} \leftarrow s_{qn}}(s_{qn}) \propto \prod_{j \neq t} p(s_{qn}|d_{qjn},p_{qjn};\tau_d,\tau_p) p(s_{qn}|b_{s_{qn}},s_{qn}\in \mathcal{S}) = \sum_{s \in \mathcal S}\pi_{qtn}(s)\delta(s-s_{qn})$, where $\pi_{qtn}(s) $ is treated as a function of $\{{d}_{qjn}\}_{j \neq t}$, $\{p_{qjn}\}_{j \neq t}$, and $\alpha_{qn}(s)$. 
	
	Consider the soft outputs of module C for $\{u_{qtn}\}$. By substituting \eqref{C} into \eqref{hybrid_c}, we obtain $u_{qtn} = p_{qtn} s_{qn} + w_{c_{qtn}} s_{qn}$ with $p(u_{qtn}|p_{qtn},s_{qn};\tau_p) = \mathcal{CN}(u_{qtn};p_{qtn}s_{qn},\tau_p)$. Then, given $p(u_{qtn}|p_{qtn},s_{qn};\tau_p)$ and $\Delta_{\delta_{u_{qtn}} \leftarrow s_{qn}}(s_{qn}), n=N_{P}+1,...,N$, we express soft outputs $\{\Delta_{\delta_{u_{qtn}} \to u_{qtn}}(u_{qtn})\}_{n=N_{\rm p}+1}^{N}$ as $ \int_{s_{qn}} p(u_{qtn}|p_{qtn},s_{qn};\tau_p)$ $\Delta_{\delta_{u_{qtn}} \leftarrow s_{qn}}(s_{qn}) \! = \! \sum_{s \in \mathcal S} \pi_{qtn}(s)$ $\mathcal{CN}(u_{qtn};p_{qtn} s,\tau_p)$. Given $ p(u_{qtn}|  p_{qtn},s_{qn};\tau_p)$ with $\{s_{qn}\}_{n=1}^{N_{\rm P}}$ being pilots, we have $\Delta_{\delta_{u_{qtn}} \to u_{qtn}}(u_{qtn}) = \mathcal{CN}(u_{qtn} ; p_{qtn} s_{qn},\tau_p), n=1,...,N_{\rm p}$. Combining $\Delta_{\delta_{u_{qtn}} \to u_{qtn}}(u_{qtn})$ and \eqref{D}, the state equation at super variable node $\{ u_{qtn} \}$ is expressed as
	\begin{align}
	 {v}_u = & \frac{N-N_{\rm P}}{N} \mathbb{E} [{\rm var}(u_t|d_{t},p_{t},s \sim \pi;\tau_{d},\tau_{p}) ] \notag \\ 
	  & + \frac{N_{\rm P}}{N} \mathbb{E} [{\rm var}(u_t|d_{t},p_{t};\tau_{d},\tau_{p}) ]  \label{v_u} 
	\end{align}
	which is the weighted sum of the estimation MSE of $\{u_{qtn}\}$, where the first term is the estimation MSE of $\{u_{qtn}\}_{n=N_{\rm P}+1}^{N}$ with $\{s_{qn}\}_{n=N_{\rm P}+1}^{N}$ being random variables, and the second term is that of $\{u_{qtn}\}_{n=1}^{N_{\rm P}}$ with $\{s_{qn}\}_{n=1}^{N_{\rm P}}$ being pilots.

	State variable $v_c$ at super variable node $\{ c_{qtn} \}$ is obtained similarly to \eqref{v_u}. By substituting \eqref{D} into \eqref{hybrid_c}, we obtain $c_{qtn}= d_{qtn}/s_{qn} + w_{d_{qtn}}/s_{qn}$ with $p(c_{qtn}|d_{qtn},s_{qn};\tau_d)=\mathcal{CN}(c_{qtn};d_{qtn}/s_{qn},\tau_d)$. Given $p(c_{qtn}|d_{qtn},s_{qn};\tau_d)$ and $\Delta_{\delta_{u_{qtn}} \leftarrow s_{qn}}(s_{qn})$, we have $\Delta_{\delta_{u_{qtn}} \to c_{qtn}}(c_{qtn}) = \sum_{s \in \mathcal S} \pi_{qtn}(s) \mathcal{CN} (c_{qtn};$ $d_{qtn}/s,\tau_d), n=N_{\rm P}+1,...,N$. Given  $p(c_{qtn}|d_{qtn},s_{qn};\tau_d)$ with $\{s_{qn}\}_{n=1}^{N_{\rm P}}$ being pilots, we have $\Delta_{\delta_{u_{qtn}} \to c_{qtn}}(c_{qtn}) = \mathcal{CN}(c_{qtn};d_{qtn}/s_{qn},\tau_d), n=1,...,N_{\rm p} $. Combining $\Delta_{\delta_{u_{qtn}} \to c_{qtn}}(c_{qtn})$ and \eqref{C}, the state equation at super variable node $\{ c_{qtn} \}$ is
	\begin{align}
	 {v}_c = & \frac{N-N_{\rm P}}{N} \mathbb{E} [{\rm var}(c_t|d_{t},p_{t},s \sim \pi;\tau_{d},\tau_{p}) ] \notag
			\\
		 & + \frac{N_{\rm P}}{N} \mathbb{E}	[{\rm var}(c_{t}|d_{t},p_{t};\tau_d,\tau_p) ].
			\label{v_c} 
	\end{align}

	To summarize, \eqref{tau_r}, \eqref{tau_po}, \eqref{v_x}, and \eqref{v_s}-\eqref{v_c} are the SE equations of Algorithm 1. 
	With initialization ${v}_x={v}_s=1$ and ${v}_c={v}_u=K/N \zeta$, the fixed point $({v}_x^*,{v}_s^*)$ of the SE equations gives the estimation MSEs of $\mathbf X$ and $\mathbf S_{\rm D}$, respectively.

	\subsection{Asymptotic Analysis}
	
	We now give a more rigorous description of the SE. Consider the large-system limit, i.e., $N$, $M$, and $K$ go to infinity with the ratios $N/M$ and $N/K$ fixed, simply denoted by $N \to \infty$. 
	Following \cite{AMP_SE}, we use the notation $ x \stackrel{\mathrm{d}}{=} y$ to represent that $\mathbb E[\phi( x)  z] = \mathbb E[\phi(y) z]$ for any integrable function $\phi(\cdot)$ and any random variable $z$.

	\begin{theorem}
	
	Consider Algorithm 1 for the SAPIT-MIMO system. Assume
	\begin{enumerate}
		\item The entries in $\mathbf G$ are i.i.d. with $g_{mn} \sim \mathcal{CN}(g_{mn};0,1/M)$, the entries in $\mathbf H$ are i.i.d. with $h_{mk} \sim \mathcal{CN}(h_{mk};0,1/M)$, and the entries in $\mathbf F$ are i.i.d. with $f_{nk} \sim \mathcal{CN}(f_{nk};0,\zeta/N)$;
		\item $b_{x_{qtk}}$ is independent of $r_{qtk}-x_{qtk}$ and $o_{qtk}-x_{qtn}$, and $b_{s_{qn}}$ is independent of $\{d_{qtn}-p_{qtn}s_{qn} \}_t$ and $\{p_{qtn}\}_t$; $\{b_{x_{qtk}}\}$ are i.i.d. and $\{b_{s_{qn}}\}$ are i.i.d..
	\end{enumerate}
	Then, we have
	\begin{enumerate}
		\item $r_{qtk}-x_{qtk}  \stackrel{d}{=} w_{r_{qtk}} + o(1)$ in \eqref{R1};
		\item $d_{qtn} - u_{qtn} \stackrel{d}{=} w_{d_{qtn}} + o(1) $ in \eqref{D};
		\item $o_{qtk}-x_{qtk}  \stackrel{d}{=} w_{o_{qtk}} + o(1) $ in \eqref{R2};
		\item $c_{qtn}-p_{qtn} \stackrel{d}{=} w_{c_{qtn}} + o(1) $ in \eqref{C}
	\end{enumerate}
	where $o(1)$ is a scalar that converges to $0$ almost surely as $N\to \infty$. Furthermore, SE equations \eqref{tau_r}, \eqref{tau_po}, \eqref{v_x}, and \eqref{v_s}-\eqref{v_c} hold almost surely as $N \to \infty$.
	\end{theorem} 
	
	\vspace{-0.4cm}
	\begin{proof}
		See Appendix A. 
	\end{proof}
	
	\vspace{-0.3cm}
	\begin{remark} 
	Theorem 1 suggests that the asymptotic performances on $\{x_{qtk}\}$ and $\{s_{qn}\}$, e.g., BER performance, can be obtained from the analysis under the AWGN channel models in \eqref{R1} and \eqref{R2} and the Rayleigh fading model in \eqref{S}. We note that Assumption 2) of Theorem 1 is made on the decoder output, and is thus not required when applying Theorem 1 to an uncoded system. For a coded system, the first part of Assumption 2) is  the independence assumption between the decoder output $b_{x_{qtk}}$ (or $b_{s_{qn}}$) 
	and the inputs $r_{qtk}-x_{qtk}$ and $o_{qtk}-x_{qtn}$ (or $\{d_{qtn}-p_{qtn}s_{qn} \}_t$ and $\{p_{qtn}\}_t$); the second part is the i.i.d assumption of $\{b_{x_{qtk}}\}$ (or $\{b_{s_{qn}}\}$). The first part can be ensured since the decoder outputs are \emph{extrinsic} messages which exclude the corresponding input messages \cite{brink}. The second part can be ensured by using random coding (such as in LDPC codes) and random interleaving techniques (such as in turbo detection \cite{Xiaojun_TIT}). From the numerical results in Sec. V, we see that SE provides a tight lower bound of the MSE performance for a coded system, which justifies the validity of Assumption 2). 
	\end{remark}

	\subsection{Achievable Rate Analysis}
	
	\begin{figure}[h]
		\centering
		\begin{minipage}{0.95\linewidth}
			\centering
			\includegraphics[width=1.0\linewidth]{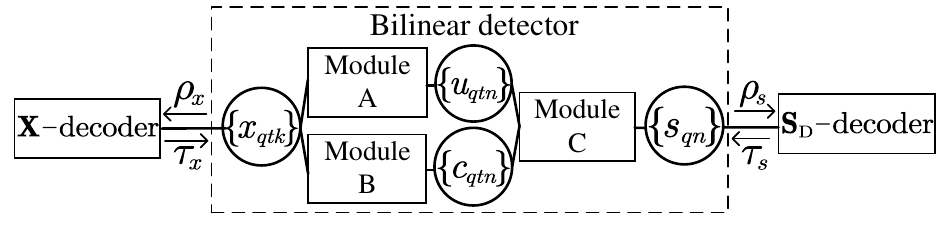}
		\end{minipage}
		\begin{minipage}{0.95\linewidth}
			\centering
			\includegraphics[width=1.0\linewidth]{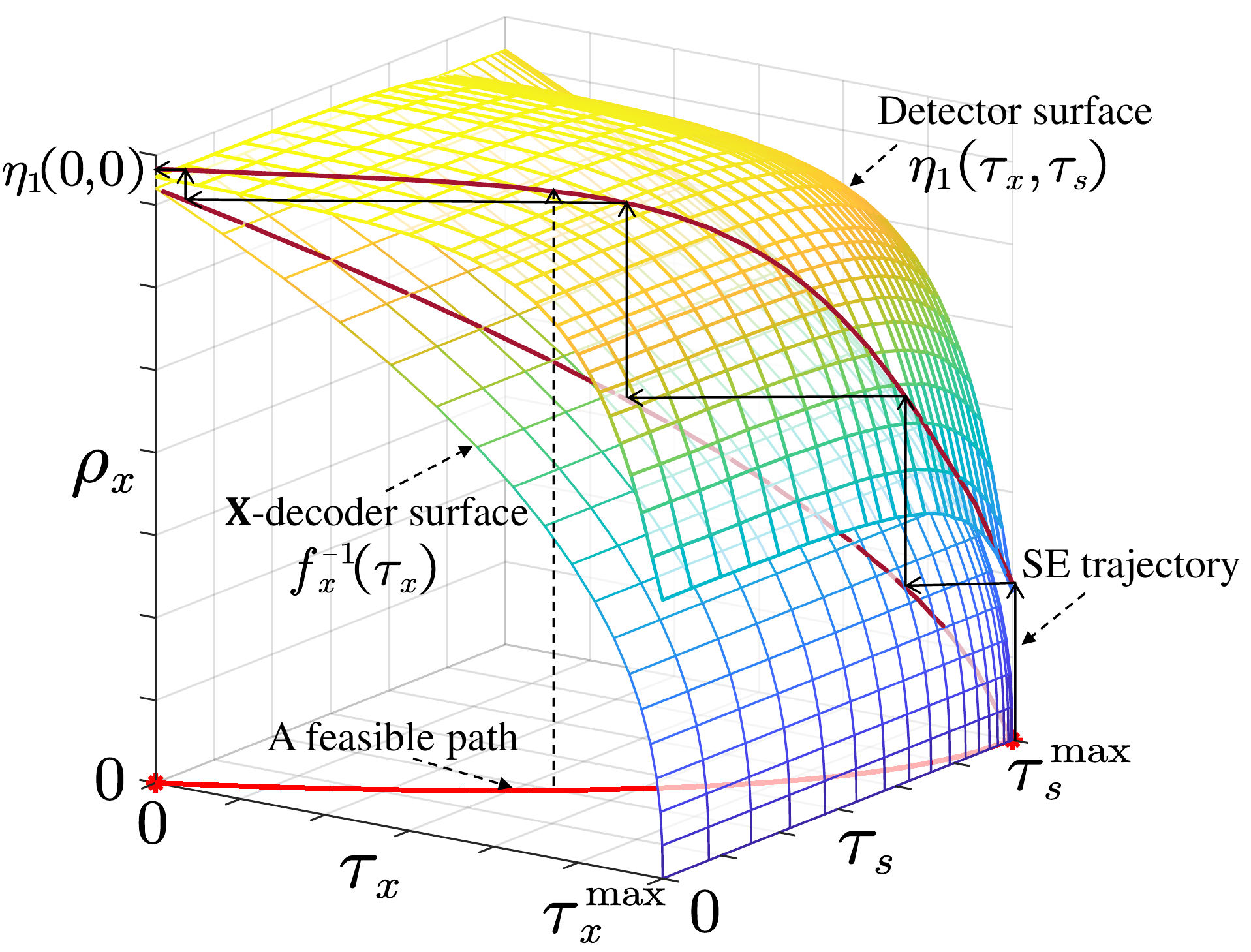}		
		\end{minipage}
		\caption{Top: an equivalent graph of Fig. \ref{App} for the achievable rate analysis. Bottom: a graphical illustration of the SE between the bilinear detector and the $\mathbf X$-decoder.} 
		\label{SE2}
	\end{figure}
	
	To facilitate the rate analysis, we combine modules A, B, and C, and all super variable nodes to form a bilinear detector of $\mathbf X$ and $\mathbf S_{\rm D}$, as illustrated in the upper plot of Fig. \ref{SE2}. Then we describe the state evolution between the bilinear detector and the two decoders as follows. At the $\mathbf X$-decoder, the combination of two input AWGN observations with variances $\tau_r$ and $\tau_o$ can be considered as one AWGN observation with SNR $\rho_x=1/\tau_r+1/\tau_o$. Furthermore, we assume that the decoder outputs $\{b_{x_{qtk}}\}$ can be characterized by a probability model with a single model parameter $\tau_x \in (0,\tau_x^{\rm max})$. $\tau_x \to 0$ represents that $\mathbf X$ is correctly detected and $\tau_x\to \tau_x^{\rm max}$ represents no information from the $\mathbf X$-decoder, i.e., $\{x_{qtk}\}$ are i.i.d. with $x_{qtk} \sim \frac{1}{|\mathcal X|} \sum_x \delta(x_{qtk}-x)$. For example, considering BPSK modulation, $b_{x_{qtk}}$ is a log-likelihood ratio with model $b_{x_{qtk}} = \frac{1}{2\tau_x}x_{qtk} + n_{x_{qtk}}, ~ n_{x_{qtk}} \sim \mathcal{CN}(0,\tau_x^{-1})$ \cite{brink}. Through $\tau_x$, we establish a single-letter characterization of the $\mathbf X$-decoder outputs. We then express the transfer function of the $\mathbf X$-decoder as
	\begin{align}
		\tau_x = f_x(\rho_x) \label{psi_x}.
	\end{align} 
	
	At the $\mathbf S_{\rm D}$-decoder, the inputs are the Rayleigh fading observations $\{d_{qtn}\}$ in \eqref{S}, where the Rayleigh model can be characterized by a single model parameter, i.e., SNR $\rho_s = \frac{\zeta K/N - \tau_p}{\tau_p+\tau_d}$. Similarly to \eqref{psi_x}, we introduce state $\tau_s \in (0 , \tau_s^{\rm max})$ for the establishment of the single-letter characterization of the $\mathbf S_{\rm D}$-decoder, and express the transfer function as
	\begin{align}
		\tau_s = f_{s}(\rho_s). \label{psi_s}
	\end{align}

	At the bilinear detector, the output state is $(\rho_x,\rho_s)$ and the input state is $( \tau_{{b}_s} , \tau_{{b}_s}) $. The corresponding input-output relationship is expressed as 
	\begin{align}
	(\rho_x,\rho_s) = \eta( \tau_x,\tau_s ),  \label{eta}
	\end{align}
	where $\eta(\cdot)$ can be determined by noting $\rho_x=\frac{1}{\tau_r}+\frac{1}{\tau_o}$ and $\rho_s = \frac{\zeta K/N - \tau_p}{\tau_p+\tau_d}$ with $(\tau_r,\tau_o,\tau_d,\tau_p)$ given by the state equations and transfer functions involved in the bilinear detector, i.e., the iteration of \eqref{tau_r}, \eqref{tau_po}, \eqref{v_x}, and \eqref{v_s}-\eqref{v_c} with ${v}_c={v}_u=K/N \zeta$. Note that $\eta(\cdot)$ is determined by ratios $N/M$, $K/M$, and $\zeta = \frac{||\mathbf G||_F^2 ||\mathbf F||_F^2}{N ||\mathbf H||_F^2}$, noise variance $\sigma_w^2$, and constraints $\mathcal X$ and $\mathcal S$. $f_x(\cdot)$ and $f_s(\cdot)$ can be adjusted through different codebooks $\mathcal C_x$ and $\mathcal C_s$. We can choose appropriate codebooks to ensure the SE convergence to $\tau_x=\tau_s=0$, i.e., to achieve the error-free communication.

	The lower plot of Fig. \ref{SE2} gives an intuitive illustration of the condition for error-free communication. Specifically, it presents the SE between the detector and the $\mathbf X$-decoder, where $\eta_i(\cdot)$ represents the $i$-th entry of vector $\eta(\cdot)$ and $f_x^{-1}(\cdot)$ is the inverse function of $f_x(\cdot)$. Note that the SE between the detector and the $\mathbf S_{\rm D}$-decoder is similar and thus is omitted. The SE trajectory in the right plot of Fig. \ref{SE2} describes how $\rho_x$ and $(\tau_x,\tau_s)$ evolve according to \eqref{psi_x}-\eqref{eta}. We see that when the decoder surface $f_x^{-1}(\cdot)$ is below the detector surface $\eta_1(\cdot)$, the SE trajectory ends at $\tau_x\to 0$ and $\tau_s\to 0$. 
	
	\begin{theorem}
		Consider the SE \eqref{psi_x}-\eqref{eta}. Suppose $f_x(\rho_x)$ and $f_s(\rho_s)$ are monotonically decreasing functions with $\tau_x^{\rm max} = f_x(0)$ and $\tau_s^{\rm max} = f_s(0)$, respectively; further suppose $\rho_x$ and $\rho_s$ both monotonically decrease with respect to $\tau_x$ or $\tau_s$ in $\eta(\tau_x,\tau_s)$. Then, the SE \eqref{psi_x}-\eqref{eta} converges. In particular, a sufficient and necessary condition for the SE convergence $ (\tau_x,\tau_s) \to (0,0)$ is that there exists a 2-dimensional monotonic path $\mathcal L$ starting from $(\tau_x,\tau_s)=(\tau_x^{\rm max},\tau_s^{\rm max})$ and ending at $(0,0)$, such that $(f_x^{-1}(\tau_x),f_s^{-1}(\tau_s)) \prec \eta(\tau_x,\tau_s) $ for any $(\tau_x,\tau_s) \in \mathcal{L}$.
	\end{theorem}
	\begin{proof}
		Consider the following iterative steps of SE: $(\rho_x^0,\rho_s^0) \! = \! \eta(\tau_x^0,\tau_s^0) \rightarrow \tau_x^1 \! = \! f_x(\rho_x^0), \tau_s^1 = f_s(\rho_x^1) \rightarrow (\rho_x^1,\rho_s^1) \! = \! \eta(\tau_x^1,\tau_s^1) \rightarrow \cdots $, where superscripts indicate iteration number and $(\tau_x^0,\tau_s^0) = (\tau_x^{\rm max},\tau_s^{\rm max})$. In step 1, we have $(0,0) \prec (\rho_x^0,\rho_s^0) $ by noting $\rho_x=\frac{1}{\tau_r}+\frac{1}{\tau_o}$ and $\rho_s = \frac{\zeta K/N - \tau_p}{\tau_p+\tau_d}$ with $(\tau_r,\tau_o,\tau_d,\tau_p)$ given by the iteration of \eqref{tau_r}, \eqref{tau_po}, \eqref{v_x}, and \eqref{v_s}-\eqref{v_c}. In step 2, given $ (0,0)\prec (\rho_x^0,\rho_s^0) $ and the monotonicity of $f_x(\cdot)$ and $f_s(\cdot)$, we have $(\tau_x^1,\tau_s^1) = (f_x(\rho_x^0),f_s(\rho_s^0)) \prec (f_x(0),f_s(0))=(\tau_x^0,\tau_s^0)$. In step 3, given $(\tau_x^1,\tau_s^1) \prec (\tau_x^0,\tau_s^0)$ and the monotonicity of $\eta(\cdot)$, we have $(\rho_x^0,\rho_s^0) \prec (\rho_x^1,\rho_s^1)$. We now prove $(\tau_x^1,\tau_s^1) \prec (\tau_x^0,\tau_s^0)$ and $(\rho_x^0,\rho_s^0) \prec (\rho_x^1,\rho_s^1)$. The subsequent processes are similar, and the SE converges to a fixed point $(\tau_x^*,\tau_s^*,\rho_x^*,\rho_s^*)$ satisfying $\eta(\tau_x^*,\tau_s^*) = (\rho_x^*,\rho_s^*) = (f_x^{-1}(\tau_x^*),f_s^{-1}(\tau_s^*))$. 
	
		We first prove the necessity part. The SE convergence to $ (\tau_x,\tau_s) \to (0,0)$ means that there is a corresponding  SE trajectory starting from $(\tau_x^{\rm max},\tau_s^{\rm max})$ and ending at $(0,0)$. Along the SE trajectory, the $\mathbf X$-decoder surface and the $\mathbf S_{\rm D}$-decoder surface are both below the detector surface. Then, there exists one monotonic path $\mathcal L$ which contains all points $(\tau_x,\tau_s)$ on this SE trajectory, such that $(f_x^{-1}(\tau_x),f_s^{-1}(\tau_s)) \prec \eta(\tau_x,\tau_s) $ for any $(\tau_x,\tau_s) \in \mathcal{L}$.

		We now prove the sufficiency part by contradiction. Suppose that the SE converges to a fixed point with $(f_x^{-1}(\tau_x^*),f_s^{-1}(\tau_s^*)) = \eta(\tau_x^*,\tau_s^*)$ and $(0,0)\prec(\tau_x^*,\tau_s^*)$. Given one monotonic path $\mathcal L$ starting from $( \tau_x^{\rm max},\tau_s^{\rm max} )$ and ending at $(0,0)$, we select point $(\tilde{\tau}_x,\tilde{\tau}_s) \in \mathcal L $ which satisfies $\tilde{\tau}_s = \tau_s^*$. We then obtain $\tau_x^* = \tilde{\tau}_x $, $\tau_x^* < \tilde{\tau}_x $, or $\tau_x^* > \tilde{\tau}_x $. If $\tau_x^* = \tilde{\tau}_x$, then $(\tau_x^*,\tau_s^*) \in \mathcal L$ and $(f_x^{-1}(\tau_x^*),f_s^{-1}(\tau_s^*)) \prec \eta(\tau_x^*,\tau_s^*)$, which is in contradiction to the fixed point hypothesis; If $\tau_x^* < \tilde{\tau}_x $, then $f^{-1}(\tau_s^*) = f^{-1}(\tilde{\tau}_s) < \eta_2(\tilde{\tau}_x,\tilde{\tau}_s) < \eta_2(\tau_x^*,\tau_s^*)$, where the first inequality is from $(\tilde{\tau}_x,\tilde{\tau}_s) \in \mathcal L $ and the second inequality is from the monotonicity of $\eta(\cdot)$; If $\tau_x^* > \tilde{\tau}_x $, we consider another point $(\bar{\tau}_x,\bar{\tau}_s) \in \mathcal L $ which satisfies $\bar{\tau}_x = \tau_x^*$. Given $(f_x^{-1}(\tau_x^*),f_s^{-1}(\tau_s^*)) = \eta(\tau_x^*,\tau_s^*)$ and the monotonicity of $\eta(\cdot)$, we have $\bar{\tau}_s < \tau_s^*$, since otherwise $(f_x^{-1}(\bar{\tau}_x),f_s^{-1}(\bar{\tau}_s)) < \eta(\bar{\tau}_x,\bar{\tau}_s)$ is not fulfilled. Then, two points $(\bar{\tau}_x,\bar{\tau}_s) $ and $(\tilde{\tau}_x,\tilde{\tau}_s)$ belongs to $\mathcal L$ with $\bar{\tau}_x > \tilde{\tau}_x$ and $\bar{\tau}_s <\tilde{\tau}_s$, which is in contradiction to the monotonic requirement of $\mathcal L$.
	\end{proof}
	
	\vspace{-0.2cm}
	\begin{remark}
		Generally speaking, $f_x(\rho_x)$ (or $f_s(\rho_s)$) is a monotonically decreasing function \cite{Dongning3}, which means that a higher SNR $\rho_x$ (or $\rho_s$) leads to a lower estimation error represented by $\tau_x$ (or $\tau_s$). The monotonicity of $\eta(\tau_x,\tau_s)$ is from the similar fact. It is noteworthy that the number of possible 2-dimensional paths $\mathcal L$ is infinite, and error-free communication requires only one path to satisfy the above-mentioned condition.
	\end{remark}
	
	\vspace{-0.2cm}
	\begin{definition}
		We say that the transfer functions $\eta(\cdot)$, $\psi_x(\cdot)$, and $\psi_s(\cdot)$ satisfy the matching condition under $\mathcal L$ if along a $2$-dimensional monotonic path $\mathcal L$ whose coordinates start from $(\tau_x^{\rm max},\tau_s^{\rm max})$ and end at $(0,0)$, such that
		\begin{align}
			(f_x^{-1}(\tau_x),f_s^{-1}(\tau_s)) = \eta(\tau_x,\tau_s).	
		\end{align} 
	\end{definition}
	\vspace{-0.2cm}
	The matching condition is the critical condition for error-free communication. In the matching condition under $\mathcal L$, all points in $\mathcal L$ are fixed points; If $(f_x^{-1}(\tau_x),f_s^{-1}(\tau_s))\! =\! \eta(\tau_x,\tau_s) - (\Delta_x,\Delta_s)$ with $(0,0) \prec (\Delta_x,\Delta_s) $, we achieve error-free communication; If $(\Delta_x,\Delta_s) \prec (0,0)$, the SE stops at the initial point. For $f_x(\cdot)$ and $f_s(\cdot)$, the matching condition under $\mathcal L$ gives the maxima of input SNRs $\rho_x$ and $\rho_s$ given outputs $\tau_x$ and $\tau_s$, which corresponds to the worst cases of the decoder transfer functions.
	
	We now establish the rate analysis. Denote by $R_T$ and $R_R$ the information rates at the Tx and the RIS, respectively. Similarly to \eqref{psi_x}, we express $v_x$ in \eqref{v_x} as a function of SNR $\rho_x$, denoted by $v_x = \psi_x(\rho_x)$. Let ${v}_{ps}$ be the normalized estimation MSE of $\sum_t p_t {s_{qn}}$, i.e., ${v}_{ps}$ can be expressed as ${v}_{ps} = \mathbb{E}_{ \{{d}_t\}, \{{p}_t\},b_{s} } [ \sum_t \frac{|{p}_t|^2}{\zeta K/N -\tau_p}$ $\text{var}( s| \{ {d}_t \},\{ {p}_t\},s \sim p(s| b_{s} ;\tau_s) ; \tau_p,\tau_d ) ]$. Similarly to \eqref{psi_s}, we further express $v_{ps}$ as $v_{ps} = \psi_{s}(\rho_s)$. Let $\psi_{x,{\rm un}}(\rho_x)$ and $\psi_{s,{\rm un}}(\rho_s)$ be the special forms of $\psi_x(\rho_x)$ and $\psi_s(\rho_x)$ obtained for uncoded $\mathbf X$ and $\mathbf S_{\rm D}$. Note that $\psi_{x,{\rm un}}(\rho_x) \geq  \psi_x(\rho_x)$  and $\psi_{s,{\rm un}}(\rho_s)\geq  \psi_s(\rho_s)$ \cite{Lei}.
	\begin{theorem}
		For the proposed SAPIT-MIMO transceiver with SE \eqref{psi_x}-\eqref{eta}, the maximal achievable sum rate is
		\begin{align}
		\label{Rate_sum}
		 R_T \! + \!  R_R & \! 
		 = \int_0^{\rho_x^0} \! K \psi_{x,{\rm un}} (\rho_x) d \rho_x
		   + \int_{0}^{\rho_s^0} \! \frac{N \! - \! N_{\rm P}}{T} \psi_{s,{\rm un}}(\rho_s) d \rho_s  \notag \\
		& \! + \! \sup_{\mathcal{L} \in \mathcal{C}} \int_{\mathcal L}
			\! \left( \! K \psi_{x} (\eta_1(\tau_x,\tau_s))  , \frac{N \! - \! N_{\rm P}}{T}  \psi_{s}(\eta_2(\tau_x,\tau_s)) \! \right) \notag \\
		& \hspace{4.8cm}	\boldsymbol{\cdot} d \eta(\tau_x,\tau_s), 
		\end{align}	
		where $(\rho_x^0,\rho_s^0)=\eta(\tau_x^{\rm max},\tau_s^{\rm max})$; $\sup$ denotes supremum operation and $\mathcal{C}$ is the set of all $2$-dimensional monotonic paths starting from $(\tau_x^{\rm max},\tau_s^{\rm max})$ and ending at $(0,0)$; $\int_{\mathcal L} \left[ K \psi_{x} (\eta_1(\cdot)) , \frac{N-N_{\rm P}}{T}  \psi_{s}(\eta_2(\cdot)) \right] \boldsymbol{\cdot} d \eta(\cdot)$ is a line integral obtained by using the matching condition under $\mathcal L$.
	\end{theorem}
	\vspace{-0.3cm}
	\begin{proof}
		Recall that the detector outputs can be treated as the AWGN and Rayleigh-fading observations of $x_{qtk}$ and $s_{qn}$. From \cite{Dongning2}, in the AWGN model of $x_{qtk}$, achievable $R_T$ equals the area under curve $K\psi_{x}(\rho_x)$ from $\rho_x=0$ to $\eta_1(0,0)$, i.e., $R_T = K \int_{0}^{\eta_1(0,0)} \psi_x(\rho_x) d \rho_x $. We extend the above relationship between $R_T$ and $\psi_x(\rho_x)$ to the Rayleigh fading model of $s_{qn}$.
		Specifically, by regarding $p_{qtn} s_{qn}, t=1,...,T$ as a single random vector, we have $R_R = \frac{N-N_{\rm P}}{T}\int_{0}^{\eta_2(0,0)} \psi_s(\rho_s) d \rho_s $. Then, we have $R_T+R_R = K \int_{0}^{\eta_1(0,0)} \psi_x(\rho_x) d \rho_x +\frac{N-N_{\rm P}}{T}\int_{0}^{\eta_2(0,0)} \psi_s(\rho_s) d \rho_s$.
		We then divide each of the two integrals into two parts, i.e., the one from $(0,0)$ to $\eta(\tau_x^{\rm max},\tau_s^{\rm max})$, and the other from $\eta(\tau_x^{\rm max},\tau_s^{\rm max})$ to $\eta(0,0)$. In the first part, the maximum of $K \int_{0}^{\eta_1(\tau_x^{\rm max},\tau_s^{\rm max})} \psi_x(\rho_x) d \rho_x +\frac{N-N_{\rm P}}{T}\int_{0}^{\eta_2(\tau_x^{\rm max},\tau_s^{\rm max})} \psi_s(\rho_s) d \rho_s$ is obtained if $\psi_x(\cdot) = \psi_{x,{\rm un}}(\cdot)$ and $\psi_s(\cdot) = \psi_{s,{\rm un}}(\cdot)$. Correspondingly, the decoders' transfer functions satisfy $\tau_x^{\rm max} = f_x(\rho_x)$ and $\tau_s^{\rm max} = f_s(\rho_s)$ from $(\rho_x,\rho_s)=(0,0)$ to $\eta(\tau_x^{\rm max},\tau_s^{\rm max})$. In the second part, we consider the maximum of $K \int_{\eta_1(\tau_x^{\rm max},\tau_s^{\rm max})}^{\eta_1(\tau_x^{0},\tau_s^{0})} \psi_x(\rho_x) d \rho_x +\frac{N-N_{\rm P}}{T}\int_{\eta_2(\tau_x^{\rm max},\tau_s^{\rm max})}^{\eta_2(\tau_x^{0},\tau_s^{0})} \psi_s(\rho_s) d \rho_s$. From the discussion below Definition 1, this maximum is obtained by using the matching condition and considering all possible paths $\mathcal L$, which concludes the proof.  
	\end{proof}	
	\vspace{-0.2cm}
	\begin{remark} 
		Theorem 3 suggests that the matching condition is the necessary condition to achieve the maximum of $R_T + R_R$. Although the choice of $\mathcal L \in \mathcal C$ is infinite, we find from simulations that the differences of $\int_{\mathcal L} \left[ K \psi_{x} (\eta_1(\cdot)) , \frac{N-N_{\rm P}}{T}  \psi_{s}(\eta_2(\cdot)) \right] \boldsymbol{\cdot} d \eta(\cdot)$ under different paths are small. In practice, we can randomly generate some paths belonging to the set $\mathcal C$ (e.g., a straight line from $(\tau_x^{\rm max},\tau_s^{\rm max})$ to $(0,0)$) and choose the one with maximal $\int_{\mathcal L} \left[ K \psi_{x} (\eta_1(\cdot)) , \frac{N-N_{\rm P}}{T}  \psi_{s}(\eta_2(\cdot)) \right] \boldsymbol{\cdot} d \eta(\cdot)$. Along the selected path, the codes at the Tx and the Rx are then designed to fulfill $(f_x^{-1}(\tau_x),f_s^{-1}(\tau_s)) = \eta(\tau_x,\tau_s)-(\Delta_x,\Delta_s)$ with $\Delta_x$ and $\Delta_s$ taking small values. Recall that $f_x(\cdot) $ and $f_s(\cdot)$ are evaluated under the AWGN model and the Rayleigh fading model, respectively. We can borrow the tools developed in \cite{Xiaojun_TIT} and \cite{Lei} for the design of curve-matching codes.
	
	
	
	\end{remark}
	
	\vspace{-0.2cm}
	\subsection{Some Extensions}
	We now apply the SE equations to some special cases with minor modifications. Specifically, if the direct link is blocked, i.e., $\mathbf H = \mathbf 0$, we delete the computation of ${\tau}_r$, and modify ${\tau}_d=N/M {v}_u + \sigma_w^2$ in \eqref{tau_r} and ${v}_x = \mathbb{E}_{o, b_{x} } [ {\rm var}(x|o, x \sim p(x|b_{x};\tau_x) ; {\tau}_{o}) ]$ in \eqref{v_x}. In an uncoded system, we set $\alpha_{qn}(s)=1/|\mathcal S|$ and ${\beta}_{qtk}(x)=1/|\mathcal X|$ for the computation of $v_x$ in \eqref{v_x} and $v_s$ in \eqref{v_s}, respectively. For example, when BPSK modulation is adopted for $x_{qtk}$, \eqref{v_x} reduces to ${v}_x = 1 - \int_{-\infty}^{\infty} (\pi {\tau}_r)^{-1/2} \tanh (2 {\rm Re}(R)/{\tau}_r)$ ${\rm exp}(-({\rm Re}(R)-1)^2/{\tau}_r) d {\rm Re}(R)$ \cite{Lozano}. In a coded system with separate detection and decoding, we treat the system as uncoded in iterative detection, followed by conventional decoding operations. For the achievable rate calculation of the separate detection and decoding, we set $\psi_x(\cdot) = \psi_{x,{\rm un}}(\cdot)$ and $\psi_s(\cdot) = \psi_{s,{\rm un}}(\cdot)$. Then the fixed-point $(\rho_x^{*},\rho_s^{*})$ can be obtained through the SE \eqref{psi_x}-\eqref{eta} until convergence, and \eqref{Rate_sum} reduces to $R_T+R_R \leq K \int_0^{\rho_x^{*}} \psi_{x,{\rm un}}(\rho_x) + \frac{N-N_{\rm P}}{T} \int_0^{\rho_s^{*}} \psi_{s,{\rm un}}(\rho_s)  d \rho_s$. Correspondingly, the decoder transfer functions for maximal achievable sum rate should fulfill $f_x(\rho_x)=\tau_x^{\rm max}$ and $f_s(\rho_s)=\tau_s^{\rm max}$ from $(\rho_x,\rho_s)=(0,0)$ to $(\rho_x^*,\rho_s^*)$, and then sharply reduce to $f_x(\rho_x^*) \to 0$ and $f_s(\rho_s^*) \to 0$.

	\section{Numerical Results}
	\subsection{Preliminaries}
	In this section, we evaluate the performance of Algorithm 1. Consider a three-dimensional (3D) Cartesian coordinate where Tx, Rx and RIS are respectively located at $(0,0,1.5)$, $(0,500,11.5)$ and $(10, 490, 11.5)$ in meters. The large-scale fading model is $\beta_{\rm{loss}} = \beta_0 (d_{\rm {3D}} / d_0)^{-\bar{\alpha}}$, where $d_0$ denotes the reference distance, $\beta_0$ denotes the fading at the reference distance, $\bar{\alpha}$ denotes the fading exponent, and $d_{ \rm{3D}}$ denotes the distance. We set $d_0=1$ m, $\beta_0=-30$ dB, $\bar{\alpha}=2.2$ for the Tx-RIS and RIS-Rx links, and $\bar{\alpha}=3.5$ for the Tx-Rx link \cite{Wenjing_JSAC}. Then the fading $\beta_{\rm{loss}}^{\rm{Tx} \to \rm{RIS}}$, $\beta_{\rm{loss}}^{\rm{RIS} \to \rm{Rx}}$ and $\beta_{\rm{loss}}^{\rm{Tx} \to \rm{Rx}}$ respectively for the Tx-RIS, RIS-Rx, and Tx-Rx links can be calculated according to the above fading model and the 3D coordinates. We adopt the Rayleigh fading model as the small-scale fading for all channels, i.e., $g_{mn} \sim \mathcal{CN}(g_{mn} ; 0,\sqrt{\beta_{\rm{loss}}^{\rm{RIS} \to \rm{Rx}}})$,  $f_{nk} \sim \mathcal{CN}(f_{nk} ; 0,\sqrt{\beta_{\rm{loss}}^{\rm{Tx} \to \rm{RIS}}})$, and $h_{mk} \sim \mathcal{CN}(h_{mk} ; 0,\sqrt{\beta_{\rm{loss}}^{\rm{Tx} \to \rm{Rx}}})$ \cite{Shuowen}. We consider that the noise power spectrum is $-150$ dBm/Hz and the system bandwidth is $1$ MHz. 
	{Note that the RIS phase switching frequency of up to $1$ MHz can be realized by the existing RIS prototypes \cite{zhang2018space}.}
	{Unless otherwise specified, we set $N=512$, $M=512$, $K=64$, $N_{\rm P}=40$, $Q=10^3$, and $T \in\{1,2,4\}$}. 
	We consider the following baseline schemes in comparison:
	\begin{itemize}
		\item W/O RIS: No RIS is deployed, i.e., $\mathbf G = \mathbf F = \mathbf 0$ and system model \eqref{model} reduces to $\mathbf y_{qt} = \mathbf H \mathbf x_{qt} + \mathbf w_{qt}$. We use Turbo-LMMSE algorithm \cite{Xiaojun_TIT} to detect $\mathbf x_{qt}$. Similarly to Algorithm 1, Turbo-LMMSE algorithm consists of a detector and a decoder with transfer functions $\eta_{\rm turbo}(\cdot)$ and $\psi_{\rm turbo}(\cdot)$, respectively. We calculate the achievable rate of the Tx under matching condition  $\eta^{-1}_{\rm turbo}(\cdot) = \psi_{\rm turbo}(\cdot)$.
		\item Random phase: The phases of RIS elements are randomly chosen from the adjustable phase angles $\theta_{i}, i=1,...,|\mathcal{S}|$. Turbo-LMMSE is used based on system model $\mathbf y_{qt} = \mathbf B_q \mathbf x_{qt} + \mathbf w_{qt}$ where $\mathbf B_q = \mathbf G {\rm diag}(\mathbf s_q) \mathbf F + \mathbf H$ is assumed to be known by the Rx.
		\item Passive beamforming: The algorithm in \cite{Shuowen} is adopted to optimize the reflection phase, where the phase angles $\{ \theta_i \}$ are allowed to take values continuously in $[0, 2\pi]$.  Turbo-LMMSE is used at the Rx. 
		\item SAPIT with TMP: Consider the proposed SAPIT technique with TMP algorithm \cite{Wenjing_JSAC} used at the Rx. The TMP algorithm is modified to exploit the Tx codebook $\mathcal C_x$ and the RIS codebook $\mathcal C_s$.
	\end{itemize}
	\vspace{-0.3cm}
	\subsection{Performance Comparisons}
	{Fig. \ref{MSE_iter} shows the MSE performance against the iteration number of Algorithm 1, where the MSE of $\mathbf X$ (or $\mathbf S_{\rm D}$) is normalized by the total number of elements of $\mathbf X$ (or $\mathbf S_{\rm D}$).} BPSK and quadrature phase-shift keying (QPSK) are adopted at the RIS and the Tx, respectively. In the left plot of Fig. \ref{MSE_iter}, we assume that the direct link is blocked, i.e., $\mathbf H = \mathbf 0$, and that no channel encoding/decoding is used. We see that the SE results (dashed lines) coincide with the numerical simulations (solid lines), implying that the SE provides a good performance bound even for a system of moderate size. In the right plot of Fig. \ref{MSE_iter}, rate-$1/2$ convolutional codes with generator polynomials $(171,133)$ are adopted at both the Tx and the RIS, and random interleaving is applied after channel encoding. The trend in the right plot of Fig. \ref{MSE_iter} is similar to that in the left plot of Fig. \ref{MSE_iter}.
	
	\begin{figure}[h] 
		\centering
		\includegraphics[width = 3.5 in]{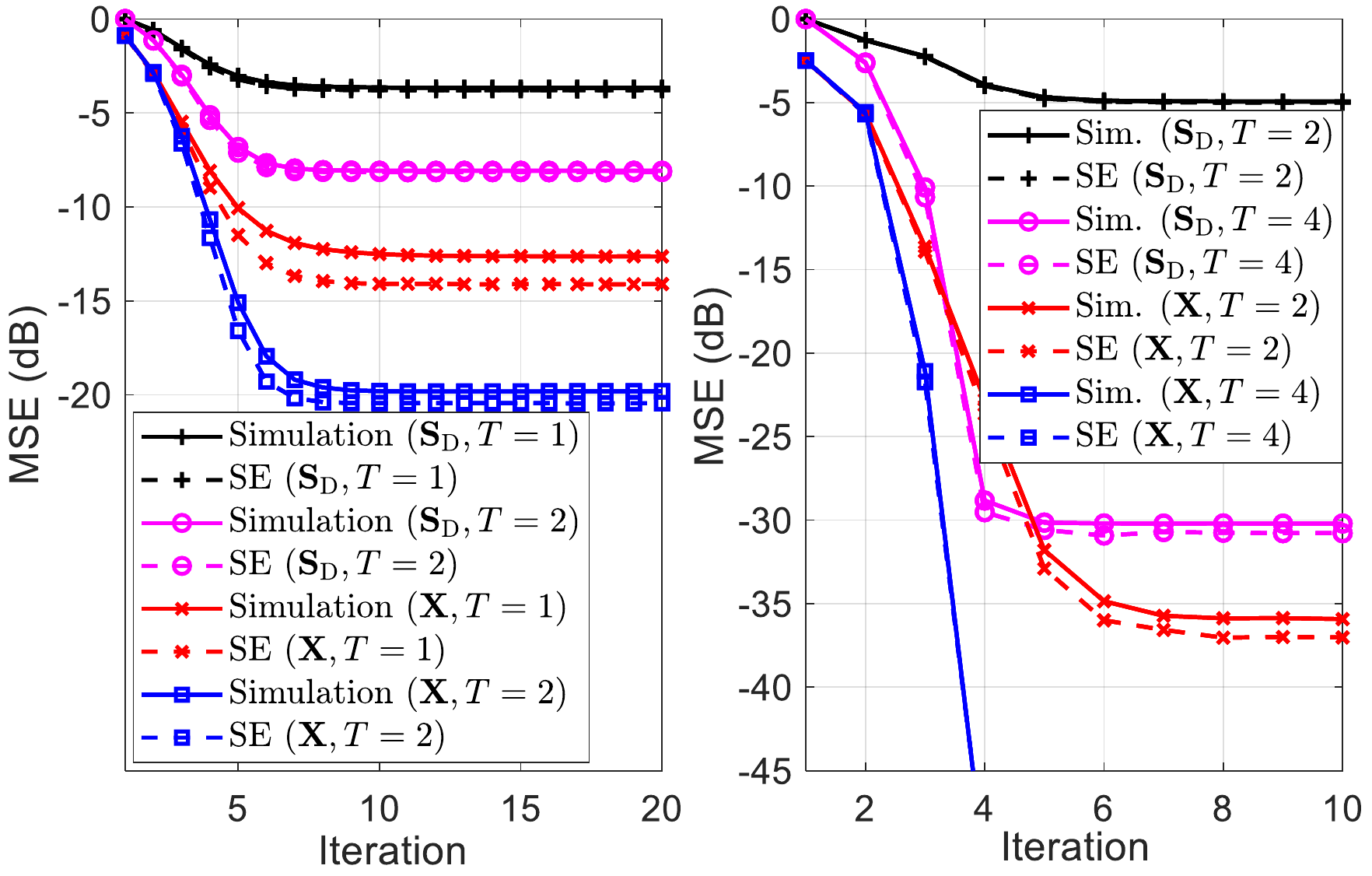}
		\vspace{-0.1cm}
		\caption{MSE versus iteration number. Left: uncoded system with transmit power $12$ dBm in the absence of direct link. Right: coded system with transmit power $6$ dBm  in the presence of direct link.}
		\label{MSE_iter}
	\end{figure}
	Fig. \ref{BER_P} shows the BER performance versus transmit power, where the simulation settings in the left and right respectively follow those in the left and right of Fig. \ref{MSE_iter}. The curve ``Oracle bound ($\mathbf X$)" means that $\mathbf S_{\rm D}$ is known by the Rx when detecting $\mathbf X$. The meaning of the curve ``Oracle bound ($\mathbf S_{\rm D}$)" is analogous. For the passive beamforming scheme, $512$ quadrature amplitude modulation (QAM) is adopted at the Tx for almost the same information rates as that of SAPIT. In the left plot of Fig. \ref{BER_P}, there is an evident gap between Algorithm 1 and TMP for the BER performances of $\mathbf X$ and $\mathbf S_{\rm D}$. Furthermore, Algorithm 1 approaches the lower bound as the Tx power increases. Passive beamforming has quite poor performance due to high-order modulation at the Tx, which confirms the superiority of SAPIT. In SAPIT, the estimates of RIS coefficients are treated as soft pilots for the estimates of the Tx signals, and vice versa. The trend in the right plot of Fig. \ref{BER_P} is similar. {We now evaluate the impact of imperfect CSI. Specifically, we consider channel estimates $\hat{\mathbf G}$, $\hat{\mathbf F}$, and $\hat{\mathbf H}$ with Gaussian errors $\Delta \mathbf G = \hat{\mathbf G} - \mathbf G$, $\Delta \mathbf F = \hat{\mathbf F} - \mathbf F$, and $\Delta \mathbf H = \hat{\mathbf H} - \mathbf H$, respectively. The normalized MSE of the channel estimation is set $-20$ dB \cite{zhenqing}, and Algorithm 1 is modified to use the channel estimates. From the right plot of Fig. \ref{BER_P}, we see that the imperfect CSI causes a slight performance loss. The reason is that it leads to the increase of the power of the equivalent noise. To see this, we rewrite (1) as $\mathbf{y}_{qt}  = \left( \hat{\mathbf G} \text{diag}(\mathbf s_q) \hat{\mathbf F} + \hat{\mathbf H} \right) \mathbf x_{qt} + \hat{\mathbf{w}}_{qt}$ with equivalent noise $\hat{\mathbf{w}}_{qt} = \mathbf w_{qt} - (\Delta \mathbf H+\hat{\mathbf G} \text{diag}(\mathbf s_q)\Delta \mathbf F + \Delta\mathbf G \text{diag}(\mathbf s_q) \mathbf F)\mathbf x_{qt} $.}
	 
	\begin{figure}[h] 
		\centering
		\includegraphics[width = 3.5in]{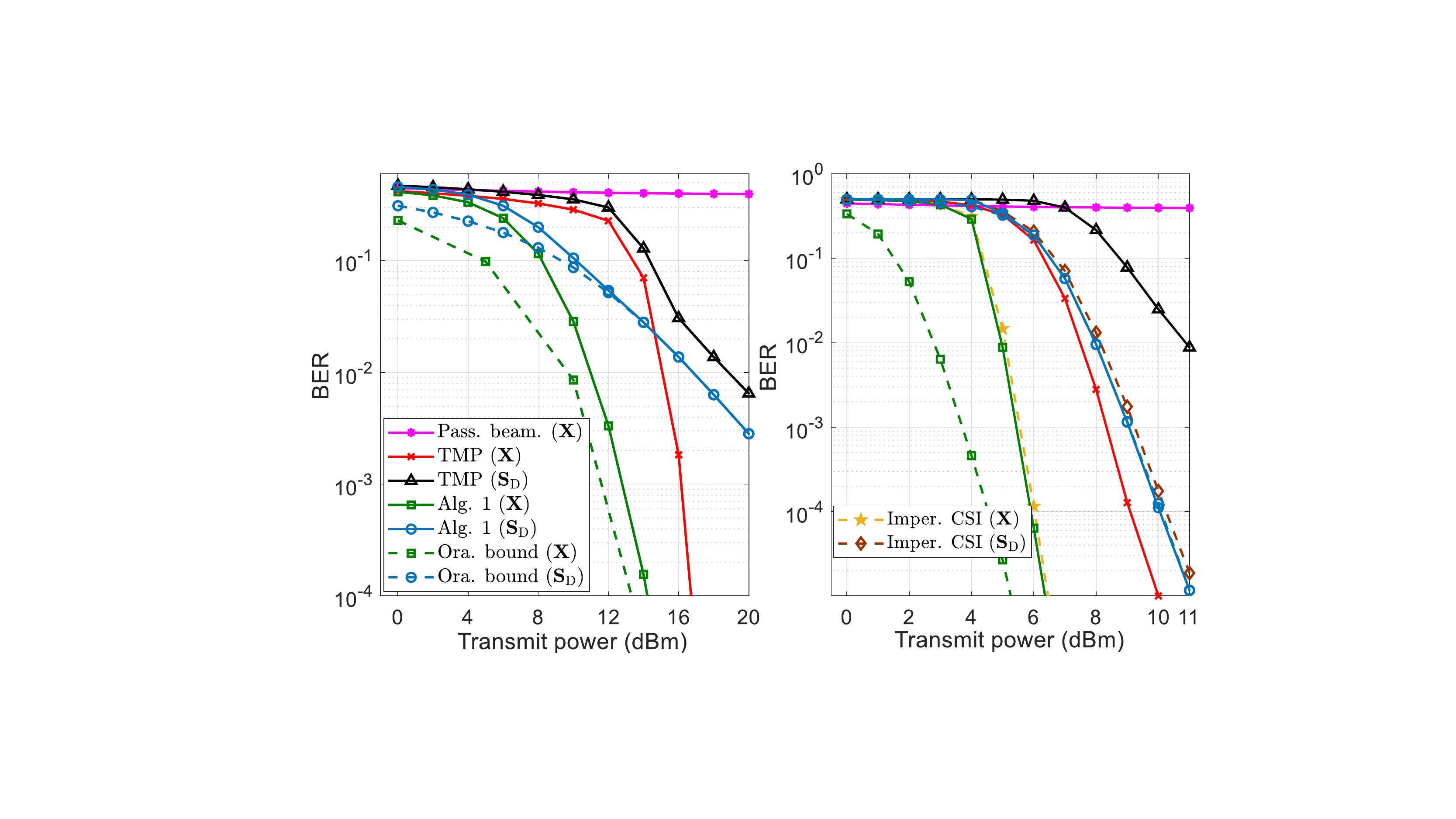}
		\vspace{-0.4cm}
		\caption{{BER versus transmit power. $T=2$. Left: uncoded system in the absence of direct link. Right: coded system with direct link.}}
		\label{BER_P}
	\end{figure}
		
	Fig. \ref{Rate_N} shows the achievable sum-rate comparisons. For the SAPIT scheme, we calculate $R_T+R_R$ by the right-hand side of \eqref{Rate_sum}, where we randomly choose $10$ curves $(\tau_x,\tau_s)$ with coordinates starting from $(\tau_x^{\rm max},\tau_s^{\rm max})$ and ending at $(0,0)$. In the left plot of Fig. \ref{Rate_N}, $16$ QAM is adopted at the Tx and $\theta \in \{ 0, \pi\}$ is adopted at the RIS. As the RIS element number increases, all schemes (except no RIS) achieves a rate improvement. Furthermore, the SAPIT scheme achieves considerable superiority over the baselines and doubles the achievable sum rate of the passing beamforming scheme at $N=800$. This superiority is also seen in the right plot of Fig. \ref{Rate_N}, where $64$ QAM and $\theta \in \{ 0, \frac{\pi}{2}, \pi, \frac{3\pi}{2} \}$ are used at the Tx and the RIS, respectively. {We note that the small rate gain by passive beamforming is due to the large sizes of the transceiver arrays, or more specifically, due to the fact that the beamforming gain by adjusting the RIS coefficients reduces as the number of spatial streams increases.}
	
	\begin{figure}[h] 
		\centering
		\includegraphics[width = 3.5in]{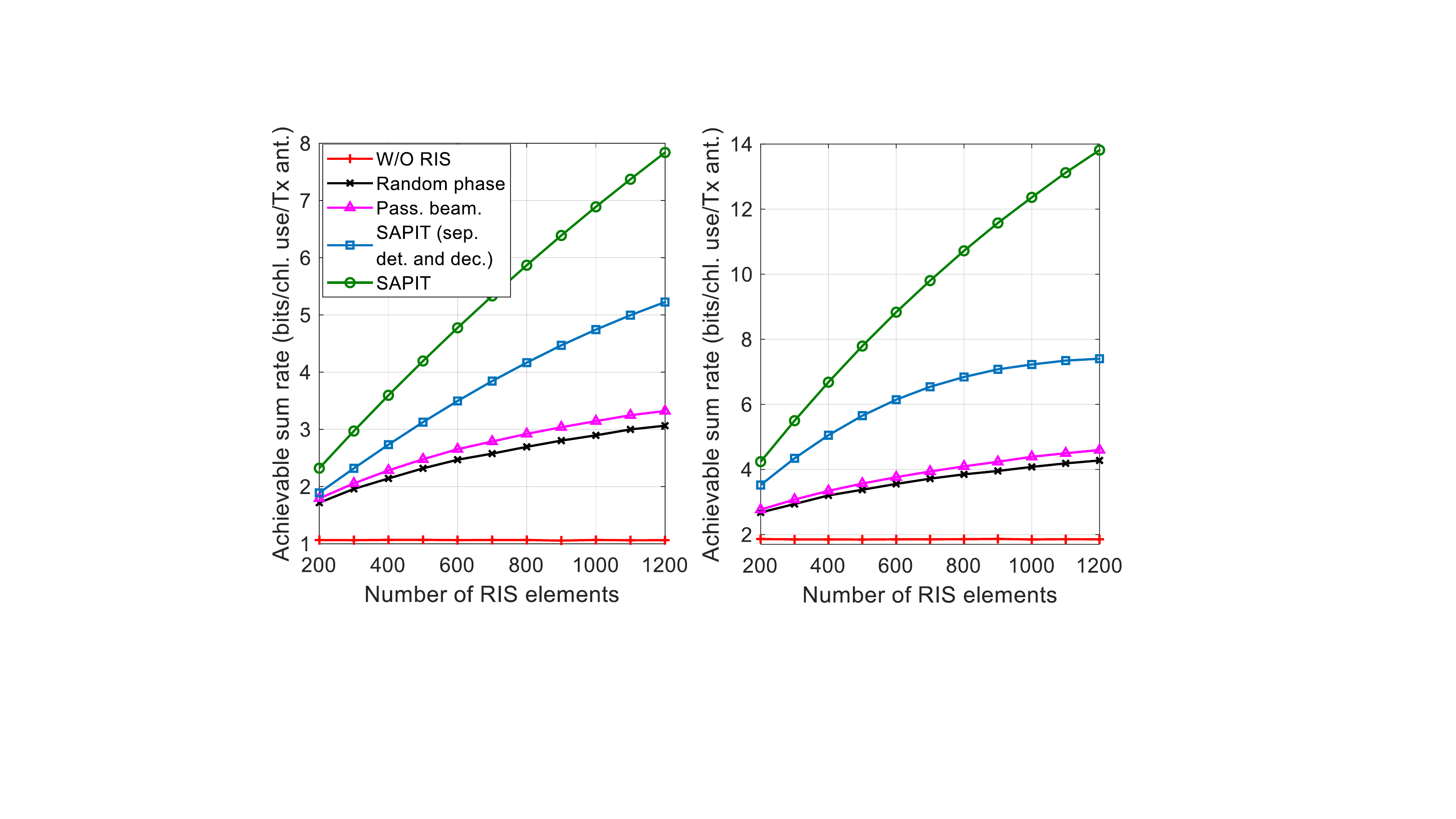}
		\caption{Achievable sum rate versus the number of RIS elements with direct link and $T=1$. Left: BPSK and 16QAM at the RIS and the Tx with transmit power $8$ dBm. Right: QPSK and 64QAM at the RIS and the Tx with transmit power $12$ dBm.}
		\label{Rate_N}
	\end{figure}

	\vspace{-0.2cm}
	\section{Conclusions}
	In this paper, we proposed a novel SAPIT transceiver in the coded RIS-aided MIMO system. We established an auxiliary system model and developed a message-passing algorithm to solve the bilinear detection problem. We further developed the SE to predict the algorithm performance. Based on the SE, we analyzed the achievable sum rate of the Tx and the RIS. Finally, we conducted simulations to validate the SE analysis and show the  superiority of the SAPIT scheme over the passive beamforming counterpart in achievable sum rate.
	
	\vspace{-0.1cm}
	\begin{appendices}

	\section{ }  
	\subsection{Preliminaries}
	
	We first introduce a general recursion as a minor modification of  \cite[eq. (83)-(85)]{sundeep}. 
	Consider an $I \times J$ random matrix $\mathbf{A} \in \mathbb{C}^{I \times J}$ with i.i.d. entries $a_{ij} \sim \mathcal{CN}(a_{ij} ;0,\zeta_a / I)$. 
	Define a set $ \{ \boldsymbol{\theta}_{i_n}(l) \in \mathbb{C}^{k_n} | n = 1,...,n_I. ~ i_n= \sum_{k=0}^{n-1} I_k + 1 ,..., \sum_{k=0}^{n} I_k.~ l=0,...,L-1. \} $ with $I_0=0$ and $\sum_{n=1}^{n_I} I_n = I$. 
	Similarly, define a set $ \{ \boldsymbol{\varphi}_{j_n}(l) \in \mathbb{C}^{q_n} | n =1,...,n_J. ~ j_n= \sum_{k=0}^{n-1} J_k + 1 ,..., \sum_{k=0}^{n} J_k.~ l=0,...,L-1. \} $ with $J_0=0$ and $\sum_{n=1}^{n_J} J_n = J$. 
	The recursion below involves the updates of $\mathbf q(l),\mathbf m_1(l),\mathbf m_2(l) \in \mathbb C^{J}$ and $\mathbf e_1(l),\mathbf e_2(l), \mathbf v(l) \in \mathbb C^{I}$. Specifically, given $\{\boldsymbol{\theta}_{i_n}(l)\}$ and $\{\boldsymbol{\varphi}_{j_n}(l)\}$, we have
	\begin{subequations}
	\label{recursion}
	\begin{align}
		\label{recursion1}
		& \mathbf q(l+1) = \mathbf A^H \mathbf v(l) - \left(\mathbf m_1(l),\mathbf m_2(l) \right) \boldsymbol{\xi}(l), \notag \\ 
		& \hspace{1pt} v_{i_n}(l) = g_{l,n} \left(e_{1i_n}(l), e_{2i_n}(l),\boldsymbol{\theta}_{i_n}(l) \right), \\
		\label{recursion2}
		& (\mathbf e_1(l), \mathbf e_2(l)) \!  = \!  \mathbf A \left(\mathbf m_1(l),\mathbf m_2(l) \right) \!\! - \!\! \mathbf v(l-1) \boldsymbol{\gamma}(l)^T, \notag \\
		& \left( m_{1j_n}(l) , m_{2j_n}(l) \right)  \!=\!  f_{l,n} \left(q_{j_n}(l), \boldsymbol{\varphi}_{j_n}(l) \right), 
	\end{align}
	\end{subequations}
	where $\boldsymbol{\xi}(l) = (\frac{\zeta_a}{I} \sum_{n,i_n} \frac{\partial v_{i_n}(l)}{\partial e_{1i_n}(l)} , \frac{\zeta_a}{I} \sum_i \frac{\partial v_{i_n}(l)}{\partial e_{2i_n}(l)})^T \in \mathbb{R}^2$, $\boldsymbol{\gamma}(l) = (\frac{\zeta_a}{I} \sum_{n,j_n} \frac{\partial m_{1j_n}(l)}{\partial q_{j_n}(l)} , \frac{\zeta_a}{I} \frac{\partial m_{2j_n}(l)}{\partial q_{j_n}(l)})^T \in \mathbb{R}^2$, and $\mathbf{v}(-1) =\mathbf 0$. \eqref{recursion} reduces to \cite[eq. (83)-(85)]{sundeep} by letting $n_I=n_J=1$, and letting $\boldsymbol{\theta}_{i_n}(l)$ and $\boldsymbol{\varphi}_{j_n}(l) $ invariant to recursion number $l$. We note that in each iteration, both $\mathbf q(l)$ and $\mathbf (\mathbf e_1(l), \mathbf e_2(l))$ are updated through a linear mixing of $\mathbf v(l)$ (or $\mathbf (\mathbf m_1(l),\mathbf m_2(l))$) by random Gaussian matrix $\mathbf A$ (or $\mathbf A^H$) together with a point-wise subtraction. The linear mixing makes $\mathbf q(l)$ (or $\mathbf (\mathbf e_1(l), \mathbf e_2(l))$) distributed as Gaussian vectors (or matrix) in the large system limit and the point-wise subtraction removes the correlation of the components of $\mathbf q(l)$ (or the rows of $\mathbf (\mathbf e_1(l), \mathbf e_2(l))$). Furthermore, due to symmetry, the components of $\mathbf q(l)$ (or the rows of $\mathbf (\mathbf e_1(l), \mathbf e_2(l))$) have the same distribution.

	
	
	To formally describe the asymptotic properties of $\mathbf q(l+1)$ and $(\mathbf e_1(l), \mathbf e_2(l))$, we introduce some definitions by following \cite{AMP_SE}. We say that a function $\phi(\cdot):\mathbb{C}^m \to  \mathbb{C}$ is \emph{pseudo-Lipschitz} of order $2$, if there exists a constant $c>0$ such that, for any $\mathbf{x}, \mathbf{y} \in \mathbb{C}^m$: $	|\phi(\mathbf x)-\phi(\mathbf y) | \leq c\left(1+\|\mathbf x\|+\|\mathbf y\| \right)\|\mathbf x - \mathbf y\|$. We say that the empirical distribution of vector sequences $\mathbf x_i, i=1,...,N$ (denoted by $\hat{p}_{\mathbf x}$) converges weakly to a probability density function $p(\tilde{\mathbf x})$ if $\lim_{N \to \infty} \mathbb E_{\hat{p}_{\mathbf x}}[\psi(x)] = E_{{p}_{\tilde{\mathbf x}}}[\psi(x)] $ for any bounded continuous function $\psi(\cdot)$. Our goal is to characterize the distribution of $\mathbf q(l+1)$ and $(\mathbf e_1(l), \mathbf e_2(l) )$ conditioned on the quantities previously calculated and used in \eqref{recursion}. To this end, define $\mathfrak{S}_{l_1, l_2}$ as the probability space of $\{\mathbf{q}(l)\}_{l=0}^{l_1-1}$,$\{\mathbf{v}(l)\}_{l=0}^{l_1-1}$,$\{\boldsymbol{\theta}_{i_n}(l)\}_{l=0}^{l_1-1}$, $\{\boldsymbol{\varphi}_{j_n}(l)\}_{l=1}^{l_2}$, $\{\mathbf{e}_1(l),\mathbf{e}_2(l)\}_{l=1}^{l_2}$, and $\{\mathbf{m}_1(l),$ $\mathbf{m}_2(l)\}_{l=0}^{l_2}$. Define matrix $\mathbf V_l =[\mathbf v(0), ..., \mathbf v(l-1)]$ and $\mathbf M_l =[\mathbf m_1(0),\mathbf m_2(0) ..., \mathbf m_1(l-1),\mathbf m_2(l-1)]$. Then we express $\mathbf v(l)$ as $\mathbf v(l) = \mathbf v_{\parallel}(l) + \mathbf v_{\perp}(l) $, where $\mathbf v_{\parallel}(l)$ is the orthogonal projection of $\mathbf v(l)$ onto the column space of $\mathbf V_l$, and $\mathbf v_{\perp}(l)$ is a vector in the orthogonal complementary space of the column space of $\mathbf V_l$. Furthermore, $\mathbf v_{\parallel}(l)$ can be  expressed as $\mathbf v_{\parallel}(l) = \sum_{i=1}^{l-1} \alpha_{i} \mathbf v(i)$ with $\alpha_{i}$ representing the $i$-th projection coefficient. Analogously to the expression of $\mathbf v(l)$, let $(\mathbf m_1(l),\mathbf m_2(l)) = (\mathbf m_{1,\parallel}(l),\mathbf m_{2,\parallel}(l) ) + (\mathbf m_{1,\perp}(l),\mathbf m_{2,\perp}(l))$ where $(\mathbf m_{1,\parallel}(l),\mathbf m_{2,\parallel}(l) ) = \sum_{i=0}^{l-1}  (\mathbf m_1(i), \mathbf m_2(i))\boldsymbol{\beta}_{i}$ with $\boldsymbol{\beta}_{i} \in \mathbb{R}^{2 \times 2}$.
	
	
	The state variables of the recursion \eqref{recursion} are $\tau_q(l) \in \mathbb R$ and $\boldsymbol{\Sigma}(l) \in \mathbb C^{2 \times 2}$ given by
	\begin{subequations}
	\begin{align}
		&\tau_q(l) = \zeta_a \sum_{n=1}^{n_I} \frac{I_n}{I} \mathbb E \left[ \left| g_{l,n} \left( \sqrt{\boldsymbol{\Sigma}(l)} \mathbf n, \tilde{\boldsymbol{\theta}}_n(l) \right) \right|^2 \right], \label{tau_q} \\
		&\boldsymbol{\Sigma}(l) = \zeta_a \sum_{n=1}^{n_J} \frac{J_n}{I}   \mathbb E \Big[ f_{l,n} (\sqrt{\tau_q(l)} w_q, \tilde{\boldsymbol{\varphi}}_{n}(l))^H \times \notag \\
		& \hspace{3cm} f_{l,n} ( \sqrt{\tau_q(l)} w_q, \tilde{\boldsymbol{\varphi}}_{n}(l) ) \Big], \label{Sigma}
	\end{align}
	\end{subequations}
	where $\mathbf n \sim \mathcal{CN}(\mathbf n ;0 , \mathbf I)$, $w_q \sim \mathcal{CN}(w_q;0,1)$, and $\boldsymbol{\Sigma}(0) = \lim_{J \to \infty} \frac{1}{J} (\mathbf m_1(0),\mathbf m_2(0))^H (\mathbf m_1(0),\mathbf m_2(0))$. The expectation $\mathbb E[\cdot]$ is taken over $\mathbf n$ and $\tilde{\boldsymbol{\theta}}_n(l)$ in \eqref{tau_q}, and $w_q$ and $\tilde{\boldsymbol{\varphi}}_{n}(l)$ in \eqref{Sigma}. 
	\begin{lemma}
		Consider the recursion \eqref{recursion}. Assume that the empirical distributions of $\boldsymbol{\theta}_{i_n}(l)$, $\boldsymbol{\varphi}_{j_n}(l)$, and $( m_{1j_n}(0), m_{2j_n}(0) )^T$ respectively converge weakly to the probability distributions of random variables $\tilde{\boldsymbol{\theta}}_n(l)$, $\tilde{\boldsymbol{\varphi}}_n(l)$, and $(\tilde{m}_{1n}(0),$ $\tilde{m}_{2n}(0))^T$ with bounded second moments. Further, assume that the empirical second moments of those vectors respectively converge to the second moments of corresponding random variables. Assume that $g_{l,n}(\cdot)$ and $f_{l,n}(\cdot)$ are Lipschitz continuous and  continuously differentiable almost everywhere with bounded derivatives. We have
	\begin{subequations}
	\begin{align}
		& \left. \mathbf q(l+1)\right|_{\mathfrak{S}_{l+1, l}} \notag \\
		& \quad \stackrel{\mathrm{d}}{=} \sum_{i=0}^{l-1} \alpha_{i} \mathbf q(i+1)+ \tilde{\mathbf A}^{H} \mathbf v_{\perp}(l) + \tilde{\mathbf M}_{l+1} \mathbf {o}_{l+1} \label{q_a} \\
		& \left. (\mathbf e_1(l),\mathbf e_2(l))\right|_{\mathfrak{S}_{l, l}}, \notag \\
		& \quad \stackrel{\mathrm{d}}{=} \sum_{i=0}^{l-1}  (\mathbf e_1(i),\mathbf e_2(i))\boldsymbol{\beta}_{i}+ \tilde{\mathbf A} (\mathbf m_{1,\perp}(l),\mathbf m_{2,\perp}(l)) + \tilde{\mathbf V}_{l} (\mathbf o_l, \mathbf o_l), \label{e_a}
	\end{align}
	\end{subequations}
	where $\tilde{ \mathbf A}$ is an independent copy of $\mathbf A$; the columns of $\tilde {\mathbf M}_{l}$ (or $\tilde {\mathbf V}_{l}$) form an orthogonal basis of the column space of $\mathbf M_l$ (or $\mathbf V_l$) with $\tilde {\mathbf M}_{l}^H \tilde {\mathbf M}_{l}= N \mathbf{I}_{l \times l}$ (or $\tilde {\mathbf V}_{l}^H \tilde {\mathbf V}_{l}= J \mathbf{I}_{l \times l}$; $\mathbf{o}_{l}$) is a vector of length $l$ whose elements converge to zero almost surely as $I, J \to \infty$. For any pseudo-Lipschitz functions $\phi_{q}(\cdot)$ and $\phi_{e}(\cdot)$ of order 2, we have
	\begin{subequations}
	\begin{align}
		& \lim _{ J_n \rightarrow \infty} \frac{1}{J_n} \sum_{j_n} \phi_{q}\left(q_{j_n}(l),  \boldsymbol{\varphi}_{j_n}(l) \right) \notag \\
		& \quad \stackrel{\text { a.s. }}{=}  \mathbb{E}\left[\phi_{q}\left( \sqrt{ \tau_q(l-1)} w_q,  \tilde{\boldsymbol{\varphi}}_n(l) \right)\right], n=1,...,n_J, \label{q_b} \\
		& \lim _{ I_n \rightarrow \infty} \frac{1}{I_n} \sum_{i_n} \phi_{e}\left(e_{1i_n}(l),e_{2i_n}(l),\boldsymbol{\theta}_{i_n}(l) \right) \notag \\
		&  \quad \stackrel{\text { a.s. }}{=}  \mathbb{E}\left[\phi_{e}\left( \sqrt{ \boldsymbol{\Sigma}(l)} \mathbf n, \tilde{\boldsymbol{\theta}}_n(l)  \right)\right], n=1,...,n_I, \label{e_b}
	\end{align}
	\end{subequations}
	where $w_q \sim \mathcal{CN}(w_q; 0,1)$ is independent of $\tilde{\boldsymbol{\varphi}}_n(l)$ and $\mathbf n \sim \mathcal{CN}(\mathbf n; 0,\mathbf I)$ is independent of $\tilde{\boldsymbol{\theta}}_n(l)$; $a\stackrel{\text { a.s. }}{=}b$ represent $a$ equals $b$ almost surely.	
	\end{lemma}
	Equations \eqref{q_a} and \eqref{q_b} in Lemma 1 mean that in the asymptotic regime $I, J \to \infty$, $\mathbf q(l)$ can be treated as a random Gaussian vector with i.i.d. entries of variance $\tau_q(l)$; \eqref{e_a} and \eqref{e_b} mean that $(\mathbf e_1(l),\mathbf e_2(l))$ can be treated as a random Gaussian matrix consisting of i.i.d. row vectors with covariance matrix $\boldsymbol{\Sigma}(l)$. The recursion \eqref{recursion} is a straightforward  extension of \cite[eq. (83)-(85)]{sundeep}, where the difference is only the choice of $n_I$, $n_J$, $\boldsymbol{\theta}_{i_n}(l)$, and $\boldsymbol{\varphi}_{j_n}(l) $. Correspondingly, Lemma 1 is an extension of \cite[Lemma 3]{sundeep}. The proof of Lemma 1 is straightforward by borrowing the methodology in \cite{sundeep}. We note that the key to ensuring Lemma 1 is the point-wise subtraction in the left hand-side of \eqref{recursion} for decorrelation and the Gaussian rotational invariance provided by the random Gaussian matrix $\mathbf A$. 
	
	\vspace{-0.2cm}
	
	\subsection{Proof}

	The results of Theorem 1 comprise parts 1) and 2) for the models of $x_{qtk}$ and $u_{qtn}$ at Module A, parts 3) and 4) for the models of $x_{qtk}$ and $c_{qtn}$ at Module B, and the SE equations including \eqref{tau_r} at Module A, \eqref{tau_po} at module B, and \eqref{v_x} and \eqref{v_s}-\eqref{v_c} at super variable nodes. 
	Note that \eqref{v_x} and \eqref{v_s}-\eqref{v_c} at super variable nodes are obtained by using the SE equations at Modules A and B.  
	Thus, it suffices to prove parts 1) and 2), and \eqref{tau_r} at module A, and parts 3) and 4), and \eqref{tau_po} at module B. We prove by showing that the message passing related to modules A and B are both special cases of recursion \eqref{recursion}, as detailed below.
	
	\subsubsection{State evolution related to module A}
	
	The message passing related to module A involves the estimates of $x_{qtk}$ and $u_{qtk}$. Under the i.i.d. assumptions (in Assumption 1 of Theorem 1) on the decoder outputs $\{b_{x_{qtk}}\}$ (or $\{b_{s_{qn}}\}$), the estimation processes of $x_{qtk}$ and $u_{qtk}$ are independent and identical at different sub-blocks $q$ and time-slot $t$. In what follows, we focus on the message passing related to module A at one time-slot in a sub-block. 
	
	Specifically, for recursion matrix, let $\mathbf A = [\mathbf G, \mathbf H] $ and $\zeta_a = 1$; for recursion vectors, let $\mathbf m_2(l)=\mathbf 0$ and $\mathbf e_2(l)= \mathbf 0$; for recursion parameters, let $n_I = 1$, $I_1=M$, and $\theta_{i_1} = w_{i_1}$ in \eqref{model}, and let $n_J = 3$, $(J_1,J_2,J_3)=(N_{\rm P},N-N_{\rm P},K)$, $\boldsymbol{\varphi}_{j_1}=(u_{j_1}, c_{j_1}, p_{j_1})$, $\boldsymbol{\varphi}_{j_2}=(u_{j_2}, c_{j_2}, p_{j_2},s_{j2})$, and $\boldsymbol{\varphi}_{j_3}=(x_{j_3},o_{j_3})$. Then the correspondence between \eqref{recursion} and the message passing related to module A is given by  
	\vspace{-0.1cm}
	\begin{subequations}
	\begin{align}
		& \mathbf q = (\mathbf u^T , \mathbf x^T)^T - ( \mathbf d^T , \mathbf r^T)^T ~ {\rm and} ~ \mathbf{e}_1 = \mathbf{w} - (\mathbf y - \mathbf b), \label{d1} \\
		& v_{i_1} = g_{1}(e_{i_1} , w_{i_1}) = e_{i_1} - w_{i_1}, \label{vn1} \\
		& m_{1j_1} 
		=  \mathbb{E}(u_{j_1}|u_{j_1}-q_{j_1},p_{j_1};{\tau}_d,{\tau}_p) - u_{j_1}, \label{mj1}  \\
		& m_{1j_2} 
		=  \mathbb{E}[u_{j_2}|u_{j_2}-q_{j_2},p_{j_2}, s_{j_2} \sim {\pi}_{j_2};{\tau}_d,{\tau}_p] - u_{j_2}, \label{mj2}  \\	
		& m_{1j_3} 
		=  \mathbb{E}[x_{j_3}|x_{j_3}-q_{j_3},o_{j_3}, x_{j_3} \sim \beta_{j3};{\tau}_{r},{\tau}_{o}] - x_{j_3}, \label{mj3}
	\end{align}
	\end{subequations}
	with $\xi_1 = 1$ and $\gamma_1 = - \left(\frac{K}{M} \sum_{j=1}^K \frac{v_{x_j}}{{\tau}_{r_1}} + \frac{N}{M} \sum_{j=1}^N \frac{v_{u_j}}{{\tau}_d} \right)  = - \frac{{\tau}_b}{({\tau}_b^{\rm p}+\sigma_w^2)}$. Note that $\mathbb E(\cdot|\cdot)$ in \eqref{mj1}-\eqref{mj3} are treated as functions, e.g., $\mathbb{E}(x_{j3}|\cdot)$ in \eqref{mj3} are functions of $q_{j3}$ and $o_{j3}$. These functions are Lipschitz continuous since the partial derivatives of these functions exist and are bounded everywhere. The corresponding functions ${\rm var}(\cdot|\cdot)$ are pseudo-Lipschitz of order $2$. For example, considering \eqref{mj3}, we have ${\rm var}(x_{j3}|\cdot) = \int |x_{j3}-\mathbb{E}(x_{j3}|\cdot)|^2 p(x_{j3}|\cdot) d x_{j3} $. Since $\mathbb{E}(x_{j3}|\cdot)$ is Lipschitz continuous and function $|\cdot|^2$ is pseudo-Lipschitz of order 2, ${\rm var}(x_{j3}|\cdot)$ belongs to pseudo-Lipschitz functions of order 2.
	
	%
	
	Applying \eqref{q_a} in Lemma 1 for $q_{j_3} = x_{j_3} - r_{1j_3}$ in \eqref{d1}, we prove part 1) of Theorem 1. Similarly, applying \eqref{q_a} in Lemma 1 for $q_{j_2} = u_{j_2} - d_{j_2}$, we prove part 2) of Theorem 1. Using \eqref{Sigma}, we obtain $\Sigma_{1,1} = \frac{N_{\rm P}}{M} \mathbb E [{\rm var}(u_{j_1}|d_{j_1},p_{j_1};{\tau}_d,{\tau}_p)] + \frac{N-N_{\rm P}}{M} \mathbb E [{\rm var} (u_{j_2}|d_{j_2},p_{j_2},s_{j_2} \sim {\pi}_{j_2};{\tau}_d,{\tau}_p)] + \frac{K}{M} \mathbb E[{\rm var}(x_{j3}|r_{j_3}, o_{j_3} , x_{j_3} \sim \beta_{j3} ;{\tau}_r ,\tau_o)] = \frac{N}{M} {v}_u + \frac{K}{M} {v}_x$. With $\Sigma_{1,1} = \frac{N}{M} {v}_u + \frac{K}{M} {v}_x$ and \eqref{tau_q}, we obtain $\tau_q = \mathbb E |\sqrt{\Sigma_{1,1}}n - w_i|^2 = \frac{N}{M} {v}_u + \frac{K}{M} {v}_x + \sigma_w^2 = \tau_r = \tau_d$ in \eqref{tau_r}.

	\subsubsection{State evolution related to module B} 
	
	Similarly to the proof in the previous subsection, let $\mathbf A = \mathbf F$, and $\zeta_a =\zeta$; $n_I = 2$, $(I_1,I_2)=(N-N_{\rm P},N_{\rm P})$, ${\boldsymbol \theta}_{i_1} = (d_{i_1}-u_{i_1}, s_{i_1}) = (w_{d_{i_1}} , s_{i_1}) $, and ${\theta}_{i_2} = d_{i_2}-u_{i_2} = w_{d_{i_2}}$; $n_J = 1$, $J_1=K$, and $\boldsymbol{\varphi}_{j_1}=(x_{j_1},r_{1j_1},{\alpha}_{j_1})$. Then the correspondence between \eqref{recursion} and the message passing related to module B is given by 
	\begin{subequations}
	\begin{align}
		& \mathbf q =  \mathbf o - \mathbf x, ~ \mathbf{e}_1 = \mathbf c, ~ {\rm and} ~  \mathbf{e}_2 = \mathbf p, \\ 
		& v_{i_1} 
		= \frac{ {\tau}_{o}}{{\tau}_p} \left( \mathbb{E}[c_{i_1}|s_{i_1} e_{1i_1} + w_{d_{i_1}},e_{2i_1}, s_{i_1} \sim {\pi}_{i_1};{\tau}_{d},{\tau}_{p}] - e_{2i_1} \right),  \\
		& v_{i_2} 
		= \frac{ {\tau}_{o}}{{\tau}_p} \left( \mathbb{E}[c_{i_2}|e_{1i_2} +w_{d_{i_2}},e_{2i_2};{\tau}_{d},{\tau}_{p}] - e_{2i_2} \right), \\
		& (m_{1j_1}, m_{2j_1}) 
		= \left(x_j, \mathbb{E}[x_{j_1}|r_{j_1},x_{j_1} + q_{j_1},x_{j_1} \sim {\alpha}_{j_1};{\tau}_{r},{\tau}_{o}] \right), 
	\end{align}
	\end{subequations}
	with $\xi_2 = \frac{\zeta {\tau}_{o}}{{\tau}_p} (\frac{1}{N}\sum_{i=1}^{N} \frac{ v_{c_i} -{\tau}_p } {{\tau}_p} ) = -1$, $\xi_1 =1$, $\gamma_2 = \frac{\zeta}{N} \sum_{j=1}^{K} \frac{v_{x_j}}{{\tau}_{o}} = \frac{{\tau}_p}{ {\tau}_{o}}$, and $\gamma_1 = 0$.

	Applying \eqref{q_a} with $q_{j_1} = o_{j_1} - x_{j_1}$, we prove part 3) of Theorem 1. Applying \eqref{e_a} with $(e_{1i},e_{2i})=(-c_i,p_i)$, we prove part 4) of Theorem 1. Using \eqref{Sigma}, we obtain $\boldsymbol{\Sigma} = \zeta \mathbb E [f_1( \sqrt{\tau_{o}}w_q ,\tilde{\boldsymbol{\varphi}}_{1})^H f_1(\sqrt{\tau_{o}}w_q , \tilde{\boldsymbol{\varphi}}_{1})] = [\zeta \frac{K}{N} , \zeta \frac{K}{N}-\zeta \frac{K}{N} {v}_x; \zeta \frac{K}{N}-\zeta \frac{K}{N} {v}_x ,\zeta \frac{K}{N} {v}_x]$, which yields $p(c_i|p_i) = \mathcal{CN}(c_i;p_i,\zeta \frac{K}{N} {v}_x)$ and $p_i \sim \mathcal{CN}(p_i;0,\zeta \frac{K}{N} - \zeta \frac{K}{N} {v}_x)$. Then we obtain the AWGN model $c_i = p_i + w_{c_i}$ in \eqref{C} with ${\tau}_p = \zeta \frac{K}{N} {v}_x$ in \eqref{tau_po}. Using \eqref{tau_q} and \cite[eq. (76)]{sundeep}, we obtain $\tau_q = \zeta \mathbb{E} \left[|g_{1}(c_i, p_i, \tilde{\boldsymbol{\theta}}_{i})|^2 \right] = \zeta {\tau}_{o}^2  \mathbb E[ \frac{\partial }{\partial p_i}(\frac{1}{\tau_p}(\frac{N_{\rm P}}{N} \mathbb E[c_i|d_i,p_i;{\tau}_d,{\tau}_p] + \frac{N-N_{\rm P}}{N} \mathbb E[c_i|d_i,p_i,s_i \sim {\pi}_i;{\tau}_d,{\tau}_p] -p_i))] = {{\tau}_p^2/}{(\zeta({\tau}_p -{v}_c))} = \tau_o $ in \eqref{tau_po}, which concludes the proof.

	\end{appendices}

	\bibliographystyle{IEEEtran}
	\bibliography{TurboMP}

	\end{document}